\documentclass[review]{elsarticle}
\usepackage{amsthm} 

\newtheorem{corollary}{Corollary}[section]

\theoremstyle{definition} 
\newtheorem{definition}{Definition}[section]
\theoremstyle{plain} 

\newtheorem{lemma}{Lemma}[section]

\newtheorem{proposition}{Proposition}[section]
\newtheorem{remark}{Remark}[section]
\newtheorem{theorem}{Theorem}[section]

\usepackage{amsmath} 
\usepackage{amssymb}

\usepackage[numbers]{natbib} 
\usepackage{slashbox} 
\usepackage{cleveref} 
\usepackage{complexity} 
\usepackage{multirow}
\Crefname{observation}{Observation}{Observations}

\usepackage{tikz}
\tikzstyle{vertex} = [circle,fill=black!0,minimum size=4pt,inner sep=1pt]
\tikzstyle{every node}=[transform shape]
\newcommand{\basisthreecycle}{
  \foreach \pos/\name in {80/a1, 100/a2, 200/b1, 220/b2, 320/c1, 340/c2}
    \path ( \pos:2.1cm) node[draw,vertex] (\name) {}; 

  \draw[gray] (a1) -- (a2);
  \draw[gray] (b1) -- (b2);
  \draw[gray] (c1) -- (c2);

  \foreach \pos/\name in {20/X1, 40/X2, 140/Y1, 160/Y2, 260/Z1, 280/Z2}
    \path ( \pos:2.1cm) node[draw,vertex] (\name) {}; 

  \draw[gray] (X1) -- (X2);
  \draw[gray] (Y1) -- (Y2);
  \draw[gray] (Z1) -- (Z2);
  \draw (X1) to [bend left=45] (X2);
  \draw (Y1) to [bend left=45] (Y2);
  \draw (Z1) to [bend left=45] (Z2);

  \foreach \u/\v in {a2/Y1, Y2/b1, b2/Z1, Z2/c1, c2/X1, X2/a1}
    \draw[dotted] (\u) -- (\v);
}

\newcommand{\basispairtwocycles}{
  \foreach \pos/\name in {35/a1, 55/a2, 125/b1, 145/b2, 215/c1, 235/c2, 305/d1, 325/d2}
    \path ( \pos:2.1cm) node[draw,vertex] (\name) {}; 

  \draw[gray] (a1) -- (a2);
  \draw[gray] (b1) -- (b2);
  \draw[gray] (c1) -- (c2);
  \draw[gray] (d1) -- (d2);

  \foreach \pos/\name in {-10/W1, 10/W2, 80/X1,100/X2,170/Y1,190/Y2,260/Z1,280/Z2}
    \path (\pos:2.1cm) node[draw,vertex] (\name) {};

  \draw[gray] (W1) -- (W2);
  \draw[gray] (X1) -- (X2);
  \draw[gray] (Y1) -- (Y2);
  \draw[gray] (Z1) -- (Z2);
  \draw (W1) to [bend left=45] (W2);
  \draw (X1) to [bend left=45] (X2);
  \draw (Y1) to [bend left=45] (Y2);
  \draw (Z1) to [bend left=45] (Z2);

  \draw[dotted] (a2) -- (X1);
  \draw[dotted] (X2) -- (b1);
  \draw[dotted] (b2) -- (Y1);
  \draw[dotted] (Y2) -- (c1);
  \draw[dotted] (c2) -- (Z1);
  \draw[dotted] (Z2) -- (d1);
  \draw[dotted] (d2) -- (W1);
  \draw[dotted] (W2) -- (a1);
}

\usepackage{pgfplots}

\usepackage{array}
\newcolumntype{S}{>{\centering\arraybackslash} m{.45\linewidth} }
\newcolumntype{U}{>{\centering\arraybackslash} m{.05\linewidth} }

\newcommand{\fallingfactorial}[2]{\ensuremath{{#1}^{\underline{#2}}}}
\newcommand{\firststirlingnumber}[2]{\ensuremath{\genfrac{[}{]}{0pt}{}{#1}{#2}}}
\newcommand{\hultmannumber}[2]{\ensuremath{\mathcal{S}_{H}(#1,#2)}}
\newcommand{\risingfactorial}[2]{\ensuremath{{#1}^{\overline{#2}}}}
\newcommand{\signedhultmannumber}[2]{\ensuremath{\mathcal{S}^{\pm}_{H}(#1,#2)}}
\newcommand{\oddhultmannumber}[2]{\ensuremath{\mathcal{S}_{H}^{odd}(#1,#2)}}

\title{The distribution of cycles in breakpoint graphs of signed permutations}
\author{Simona Grusea}
\ead{grusea@insa-toulouse.fr}
\address{Institut de Math\'ematiques de Toulouse, INSA de Toulouse, Universit\'e de Toulouse.}
\author{Anthony Labarre}
\ead{Anthony.Labarre@cs.kuleuven.be}
\address{Department of Computer Science, K. U. Leuven, Celestijnenlaan 200A - bus 2402, 3001 Heverlee, Belgium.}

\begin{document}

\begin{abstract}
\emph{Breakpoint graphs} are ubiquitous structures in the field of genome rearrangements. Their cycle decomposition has proved useful in computing and bounding many measures of (dis)similarity between genomes,
and studying the distribution of those cycles is therefore critical to gaining insight on the distributions of the genomic distances that rely on it. We extend here the work initiated by \citet{doignon-hultman}, who enumerated unsigned permutations whose breakpoint graph contains $k$ cycles, to \emph{signed} permutations, and prove explicit formulas for computing the expected value and the variance of the corresponding distributions, both in the unsigned case and in the signed case. We also show how our results can be used to derive simpler proofs of other previously known results. Finally, we compare the distribution of the number of cycles in breakpoint graphs of unsigned and signed permutations to the distributions of several well-studied genomic distances, emphasising the cases where approximations obtained in this way stand out.
\end{abstract}


\maketitle


\section{Introduction}

The field of comparative genomics is concerned with quantifying
similarity or divergence between organisms. Several measures have
been proposed to that end, including pattern matching based
approaches or edit distances relying on a given set of
biologically relevant operations. A standard example of such a
method, and a \emph{de facto} standard in phylogenetics, is
the approach based on \emph{sequence alignment}, which is motivated by the observation that
genomes evolve by point mutations and aims at explaining evolution
by replacements, insertions or deletions of single nucleotides
(see e.g.\ \citet{li-survey} for a recent account of sequence
alignment techniques and their uses).

However, genomes also evolve by large-scale mutations that act on whole segments of the genome, as opposed to point mutations. Examples of such mutations include \emph{reversals}, which reverse the order of elements along a segment, \emph{transpositions}, which move segments to another location, and \emph{translocations}, which exchange segments that belong to different chromosomes.
Many models have been proposed for studying those \emph{genome rearrangements},
which vary according to the kinds of mutations one wants to take into account,
how these should be weighted, or which objects are best suited for representing
genomes (see e.g.\ \citet{fertin-combinatorics} for an extensive survey).
Nonetheless, a striking similarity between all these models is how heavily they
rely on variants of a graph first introduced by \citet{bafna-genome-journal},
known as the \emph{breakpoint graph}, and its decomposition into edge- or
vertex-disjoint cycles, which has proved most useful in obtaining extremely
tight bounds on many genome rearrangement distances, as well as formulas for
computing the exact distance in several cases. The link between several genomic distances and the
number of cycles in breakpoint graphs will be discussed in more detail in \Cref{sec:distributions}.

Many mathematical questions arise when studying genome
rearrangement distances, particularly concerning their
distributions, as well as related statistical parameters. Since
quite a few such distances can be computed or approximated using the cycle
decomposition of the breakpoint graph, investigating the
distribution of such cycles appears as a natural, general and effective
starting point to answering those questions. We will restrict our attention in this paper to the permutation model, which can be used when all genomes under comparison consist of exactly the same genes (but in a different order) without duplications. Breakpoint graphs can
be associated to permutations,
and the distribution of cycles in
this case was first characterised by \citet{doignon-hultman},
which later led \citet{bona-average} to prove a very simple
expression for the expected value of the \emph{block-interchange
distance} originally introduced by \citet{christie-block}.

However, it has often been argued that \emph{signed permutations}
provide a more realistic model of evolution, since signs can be
used to represent on which strand a given DNA segment is located.
Using this model, \citet{szekely-expectation} obtained bounds for the expectation
and the variance of the number of cycles in the breakpoint graph
of a random signed permutation. Using the finite Markov chain
embedding technique, \citet{grusea-random} obtained the
distribution of the number of cycles in the breakpoint graph of a
random signed permutation in the form of a product of transition
probability matrices of a certain finite Markov chain.
Her method allows to derive recurrence formulas and to compute this distribution numerically, but the
computational complexity is quite high and limits the practical
applications.

In this work, we obtain a new expression for computing the number
of unsigned permutations whose breakpoint graph contains a given
number of cycles, as well as what is to the best of our knowledge
the first analytic expression for computing the number of
\emph{signed} permutations whose breakpoint graph contains a given
number of cycles. The formula obtained in the signed case is
complicated, but we obtain simpler formulas for a couple of
restricted cases. We also use our results to derive elementary proofs of
previously known results, including a binomial identity and the
distribution of the number of cycles in the breakpoint graph of an unsigned
permutation.
We prove formulas for computing the expected value and the variance of the distribution
of those cycles, both in the unsigned case and in the signed case. Finally, we also discuss how the results we obtain relate
to a number of widely-studied genome rearrangement distances, and
in particular, how the distribution of cycles in breakpoint graphs
can be used to approximate (and in some cases, to recover exactly)
the distribution of those distances.


\section{Notations and definitions}

We recall here a few notions that will be used throughout the
paper. We assume the reader is familiar with graph theory (if not,
see e.g.\ \citet{diestel-graph}), but nevertheless review a few
useful definitions, if only to agree on notation. We will work
with \emph{non-simple} graphs, i.e.\ graphs that may contain
\emph{loops} (edges connecting a vertex to itself)
as well as parallel edges. We will also
work with both undirected and directed graphs, using $\{u,v\}$ to
denote edges in the former case and $(u,v)$ to denote arcs in the
latter.

\begin{definition}
 A \emph{matching} $M$ in a graph $G=(V,E)$ is a subset of pairwise vertex-disjoint edges of $E$. It is a
\emph{perfect matching}
of $U\subseteq V$ if every vertex in $U$ is incident to an edge in $M$.
\end{definition}

\begin{definition}
 A graph is \emph{$k$-regular} if each of its vertices has degree $k$.
\end{definition}

In particular, if $G$ is a $2$-regular graph, then it decomposes
in a unique way into a collection of edge- and vertex-disjoint
cycles, up to the ordering of cycles and to rotations of elements
within each cycle (i.e., $(a,b,c,d)=(b,c,d,a)$), as well as
directions in which cycles are traversed if $G$ is undirected (i.e.,
$(a,b,c,d)=(d,c,b,a)$). This allows us to denote unambiguously
$c(G)$ the number of cycles in $G$. The \emph{length} of a
cycle\label{def:length-of-a-cycle} is the number of vertices it
contains, and a \emph{$k$-cycle} in $G$ is a cycle of length $k$.

\begin{definition}
 A graph is \emph{hamiltonian} if it contains a cycle visiting every vertex exactly once.
\end{definition}

We now recall a few basic notions about permutations (for more details, see e.g.\ \citet{bjorner-combinatorics} and \citet{wielandt-finite}).

\begin{definition}\label{def:permutation}
A \emph{permutation} of $\{1, 2, \ldots, n\}$ is a bijective application of $\{1, 2$, $\ldots$, $n\}$ onto itself.
\end{definition}

The \emph{symmetric group} $S_n$ is the set of all permutations of $\{1, 2, \ldots, n\}$, together with the usual function composition $\circ$, applied from right to left. We use lower case Greek letters to denote permutations, typically $\pi=\langle\pi_1\ \pi_2\ \cdots\ \pi_n\rangle$, with $\pi_i=\pi(i)$, and in particular write the \emph{identity permutation} as $\iota=\langle 1\ 2\ \cdots\ n\rangle$.

\begin{definition}
The \emph{graph} $\Gamma(\pi)$ of a permutation $\pi\in S_n$ has vertex set $\{1,2,\ldots$, $n\}$, and  contains an arc $(i,j)$ whenever $\pi_i=j$.
\end{definition}

\Cref{def:permutation} implies that $\Gamma(\pi)$ is $2$-regular and as such decomposes in a unique way into disjoint cycles (up to the ordering of cycles and to rotations of elements within each cycle),
which we refer to as the \emph{disjoint cycle decomposition} of $\pi$.
It is also common to refer to a permutation as a $k$-cycle, if the only cycle of length greater than $1$
that its graph contains has length $k$. \Cref{fig:exemple-graph-of-a-permutation} shows an example of such a decomposition. To lighten the presentation, we will shorten the notation $c(\Gamma(\pi))$ into $c(\pi)$, for a given permutation $\pi$.

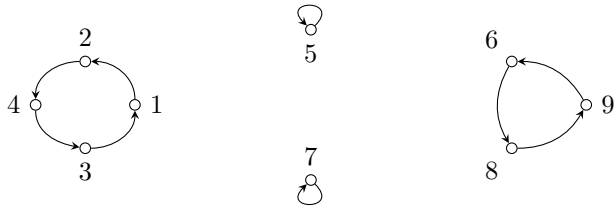
\begin{figure}[htbp]
  \centering
    \begin{tikzpicture}[scale=1,>=stealth,bend angle=45]
      \tikzstyle{arc} = [draw,->]
      \tikzstyle{every loop} = []
        \node[vertex,draw] (1) at (.66,0) [label=right:$1$] {};
        \node[vertex,draw] (2) at (0,.577) [label=above:$2$] {};
        \node[vertex,draw] (3) at (-.66,0) [label=left:$4$] {};
        \node[vertex,draw] (4) at (0,-.577) [label=below:$3$] {};
        \draw[arc] (1) to [out=90,in=0] (2);
        \draw[arc] (2) to [out=180,in=90] (3);
        \draw[arc] (3) to [out=270,in=-190] (4);
        \draw[arc] (4) to [out=0,in=-90] (1);
      \begin{scope}[xshift=3cm]
        \node[vertex,draw] (5) at (0,1) [label=below:$5$] {};
        \draw[arc] (5) to [out=45,in=135,loop] (5);
        \node[vertex,draw] (7) at (0,-1) [label=above:$7$] {};
        \draw[arc] (7) to [out=-45,in=-135,loop] (7);
      \end{scope}
      \begin{scope}[xshift=6cm]
        \node[vertex,draw] (U3) at (.66,0) [label=right:$9$] {};
        \node[vertex,draw] (U6) at (-.33,.577) [label=above left:$6$] {};
        \node[vertex,draw] (U7) at (-.33,-.577) [label=below left:$8$] {};
        \draw[arc] (U3) to [out=120,in=0] (U6);
        \draw[arc] (U6) to [out=-120,in=120] (U7);
        \draw[arc] (U7) to [out=0,in=-120] (U3);
      \end{scope}
    \end{tikzpicture}
  \caption{The graph of the permutation $\pi=\langle 2\ 4\ 1\ 3\ 5\ 8\ 7\ 9\ 6\rangle$.}
  \label{fig:exemple-graph-of-a-permutation}
\end{figure}

\begin{definition}
The \emph{conjugate} of a permutation $\pi$ by a permutation $\sigma$, both in $S_n$, is the permutation $\sigma\circ\pi\circ\sigma^{-1}$, and can be obtained by replacing every element $i$ in the disjoint cycle decomposition of $\pi$ with $\sigma_i$.
\end{definition}

\begin{definition}
A \emph{signed permutation} is a permutation of $\{1, 2, \ldots, n\}$ where each element has an additional ``$+$'' or ``$-$'' sign.
\end{definition}

The \emph{hyperoctahedral group} $S^{\pm}_n$ is the set of all signed permutations of $n$ elements, together with the usual function composition $\circ$, applied from right to left.
It is not mandatory for a signed permutation to have negative elements, so $S_n\subset S^{\pm}_n$ since each permutation in $S_n$ can be viewed as a signed permutation without negative elements. To lighten the presentation, we will conform to the tradition of omitting ``$+$'' signs for positive elements.

Finally, we
recall the definition of
the following graph introduced by \citet{bafna-genome-journal}, which turned out to be an extremely useful tool for studying and solving genome rearrangement problems and which will be central to our discussions.

\begin{definition}\label{def:BG}
Given a signed permutation $\pi$ in $S^\pm_n$, transform it into an unsigned permutation $\pi'$ in $S_{2n}$
by mapping $\pi_i$ onto the sequence $(2\pi_i-1,2\pi_i)$ if $\pi_i>0$, or $(2|\pi_i|,2|\pi_i|-1)$ if $\pi_i<0$, for $1\leq i\leq n$. The \emph{breakpoint graph} of $\pi$ is the undirected bicoloured graph $BG(\pi)$ with ordered vertex set $(\pi'_0=0,\pi'_1, \pi'_2, \ldots, \pi'_{2n},\pi'_{2n+1}=2n+1)$ and whose edge set is the union of the following two
perfect matchings
of $V(BG(\pi))$:
\begin{itemize}
\item black edges $\delta_B(\pi)=\{\{\pi'_{2i}, \pi'_{2i+1}\}\ |\ 0\leq i\leq n\}$;
\item grey edges $\delta_G=\{\{\pi'_{2i}, \pi'_{2i}+1\}\ |\ 0\leq i\leq n\}=\{\{2i, 2i+1\}\ |\ 0\leq i\leq n\}$.
\end{itemize}
We will often use the notation $BG(\pi)=\delta_B(\pi) \cup \delta_G$ to denote breakpoint graphs.
\end{definition}

Genome rearrangement problems usually involve computing edit
distances, i.e.\ the smallest number of moves needed to transform
a genome into another one using only operations specified by a
given set $S$. In the case of permutations, those distances are
usually \emph{left-invariant}, which intuitively means that genes
can be relabelled so that either genome becomes $\iota$ without
affecting the value of the distance to compute. Under this assumption, the
pairwise genome rearrangement problem in $S_n^\pm$ can be viewed
as a constrained sorting problem, and the intuition behind the
breakpoint graph construction is that black edges are meant to
represent the current situation (i.e.\ the ordering provided by
$\pi$), while grey edges are meant to represent the target
situation (i.e.\ the ordering provided by $\iota$).
\Cref{fig:breakpoint-graph-example} shows an example of a
breakpoint graph. By definition, such a graph is a collection of
even-length cycles that alternate black and grey edges. It can be
easily seen that the example shown in
\Cref{fig:breakpoint-graph-example} decomposes into two such
cycles.

\begin{figure}[htbp]
\centering
\begin{tikzpicture}
  \foreach \position/\name in {3/0, 2/10, 1/9, 0/1, 15/2, 14/3, 13/4, 12/7, 11/8, 10/14, 9/13, 8/6, 7/5, 6/11, 5/12, 4/15}
     \path ( \position*22.5:2.5cm) node {\small $\name$};

  \foreach \position/\name in {3/0, 2/1, 1/2, 0/3, 15/4, 14/5, 13/6, 12/7, 11/8, 10/9, 9/10, 8/11, 7/12, 6/13, 5/14, 4/15}
     \path ( \position*22.5:3cm) node[color=gray] {\small $\pi'_{\name}$};

  \foreach \position/\name in {3/v0, 2/v1, 1/v2, 0/v3, 15/v4, 14/v5, 13/v6, 12/v7, 11/v8, 10/v9, 9/v10, 8/v11, 7/v12, 6/v13, 5/v14, 4/v15} {
      \path (\position*22.5:2.1cm) node[draw,vertex] (\name) {}; 
  }
  \foreach \u/\v in {v0/v1, v2/v3, v4/v5, v6/v7, v8/v9, v10/v11, v12/v13, v14/v15}
      \draw (\u) -- (\v);
  \draw[gray] (v0) to [bend right=45] (v3);
  \draw[gray] (v4) to [bend right=45] (v5);
  \draw[gray] (v6) to [bend right=45] (v12);
  \draw[gray] (v7) to [bend right=45] (v11);
  \draw[gray] (v2) to [bend right=25] (v8);
  \draw[gray] (v1) -- (v13);
  \draw[gray] (v9) to [bend right=10] (v15);
  \draw[gray] (v10) to [bend right=10] (v14);
\end{tikzpicture}
\caption{The breakpoint graph of $\langle -5\ 1\ 2\ 4\ -7\ -3\ 6\rangle$.}
\label{fig:breakpoint-graph-example}
\end{figure}
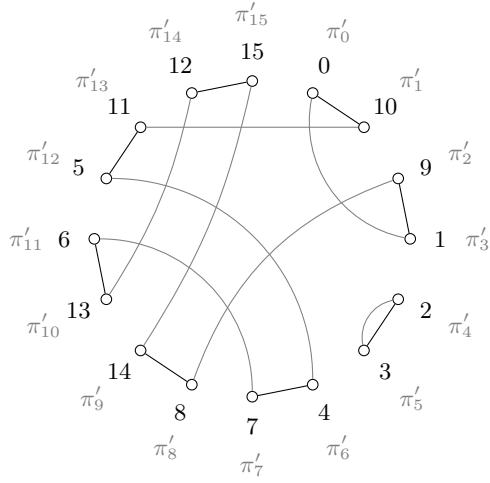

The \emph{length} of a cycle in a breakpoint graph differs from the traditional graph-theoretical definition that we mentioned on page~\pageref{def:length-of-a-cycle}:
it is \emph{half} the number of edges the cycle contains. Nevertheless, we will keep the terminology \emph{$k$-cycle} to designate a cycle of length $k$, keeping in mind that its length is measured differently in the context of breakpoint graphs.

\section{Cycle statistics}

As is well-known (see e.g.\ \citet{graham-concrete}),
the \emph{unsigned Stirling number of the first kind}
$\firststirlingnumber{n}{k}$ counts the number of permutations in
$S_n$ which decompose into $k$ disjoint cycles:
$$
\firststirlingnumber{n}{k}=|\{\pi\in S_n\ |\ c(\pi)=k\}|.
$$
Recall also that those numbers arise as coefficients in the series expansion of the \emph{rising factorial}
\begin{equation}\label{risingfactorial}
\risingfactorial{x}{n}=x(x+1)\cdots(x+n-1)=\sum_{k=0}^n\firststirlingnumber{n}{k}x^k
\end{equation}
and of the \emph{falling factorial}
\begin{equation}\label{fallingfactorial}
\fallingfactorial{x}{n}=x(x-1)\cdots(x-n+1)=\sum_{k=0}^n(-1)^{n-k}\firststirlingnumber{n}{k}x^k.
\end{equation}
Signing the elements of a permutation does not change its disjoint cycle decomposition, so the number of \emph{signed} permutations that decompose into $k$ disjoint cycles is $2^n\firststirlingnumber{n}{k}$.
We are interested in the following analogues of the Stirling number of the first kind, based on the cycle decomposition of the breakpoint graph.

\begin{definition}
The \emph{Hultman number} $\hultmannumber{n}{k}$ counts the number of permutations in $S_n$ whose breakpoint graph decomposes into $k$ cycles:
$$
\hultmannumber{n}{k}=|\{\pi\in S_n\ |\ c(BG(\pi))=k\}|.
$$
The \emph{signed Hultman number} $\signedhultmannumber{n}{k}$ counts the number of permutations in $S^{\pm}_n$ whose breakpoint graph decomposes into $k$ cycles:
$$
\signedhultmannumber{n}{k}=|\{\pi\in S^{\pm}_n\ |\ c(BG(\pi))=k\}|.
$$
\end{definition}

It is clear from \Cref{def:BG} that the number of cycles in any breakpoint graph is at least one and at most $n+1$. Hultman numbers were so named by \citet{doignon-hultman} after Axel Hultman, who first raised the question of computing those numbers~\cite{hultman-toric}. The authors obtained an explicit but complicated formula for computing $\hultmannumber{n}{k}$, as well as formulas for enumerating permutations with a given ``Hultman class'' (the analogue of conjugacy classes of $S_n$ based on the breakpoint graph).
\citet{bona-average} later
observed that they can be computed using
the following much simpler expression:
\begin{equation}\label{eqn:hultman-numbers-bona-flynn}
\hultmannumber{n}{k}=\left\{
\begin{array}{ll}
\firststirlingnumber{n+2}{k}/\binom{n+2}{2} & \mbox{if $n-k$ is odd,}\\
0 & \mbox{otherwise,}
\end{array}
\right.
\end{equation}
based on a formula first obtained by \citet{kwak-genus}.

In
the next section,
we present another way of
obtaining an explicit formula for the unsigned Hultman numbers,
which we will use in \Cref{sec:simpler-proofs} to derive a new and simple proof of
\Cref{eqn:hultman-numbers-bona-flynn}.
In \Cref{sec:formula-signed-hultman}, we will prove the first explicit formula for computing the \emph{signed} Hultman numbers.


\section{A new formula for $\hultmannumber{n}{k}$}\label{sec:new-formula-hultman}

We will need the following results obtained by
\citet{hanlon-spectra}, whose notation we follow. For any fixed $n$ in $\mathbb{N}_0$, let
$$
Q_{n}^{\mathbb{C}}(h,\ell)=\mathbb{E}(\mathrm{Re}(\mathrm{tr}((VV^{\mathrm{t}})^{n}))),
$$
where $V$ is a random $h\times \ell$ matrix with independent
standard complex normal entries, $\mathbb{E}$ denotes expectation,
$\mathrm{Re}$ denotes real part, $\mathrm{tr}$ denotes trace and
$^{\mathrm{t}}$ denotes matrix transposition. For the definition
and the properties of the complex normal distribution, see for
example \citet{Goodman}.

\citet{hanlon-spectra} give two formulas for computing
$Q_{n}^{\mathbb{C}}(h,\ell)$, both of which we will need. The
first formula\footnote{See Corollary 2.4  p.~158 of
\citet{hanlon-spectra}.} is:
\begin{equation}
Q_{n}^{\mathbb{C}}(h,\ell)=\sum_{\omega\in
S_{n}}h^{c(\omega)}\ell^{c(\omega
\circ\omega_{(n)})}, \label{formula1}%
\end{equation}
where $\omega_{(n)}$ is a fixed $n$-cycle in $S_n$.
The second formula\footnote{See Theorem 2.5 p.~158 of \citet{hanlon-spectra}.} is:
\begin{equation}
Q_{n}^{\mathbb{C}}(h,\ell)=\frac{1}{n}\sum_{i=1}^{n}(-1)^{i-1}\frac
{\fallingfactorial{(h+n-i)}{n}\fallingfactorial{(\ell+n-i)}{n}}{(n-i)!(i-1)!}. \label{formula2}%
\end{equation}

The link between the Hultman numbers and the previous results of
\citet{hanlon-spectra} is obtained using the following result of
\citet{doignon-hultman}.

\begin{corollary}\label{unsigned-hultman-counts-factorisations}\cite{doignon-hultman}
$\hultmannumber{n}{k}$ counts the number of factorisations of a
fixed $(n+1)$-cycle $\beta$ into the product $\rho\circ\omega$,
where $\rho$ is an $(n+1)$-cycle and $\omega$ a permutation in
$S_{n+1}$ with $c(\omega)=k$.
\end{corollary}

For a polynomial $P(x)$, let $[x^{k}]P(x)$ denote the coefficient of the monomial $x^{k}$ in $P(x)$. We derive the following new expression for computing $\hultmannumber{n}{k}$.

\begin{theorem}\label{new-formula-for-hultman-numbers} For all $n$ in $\mathbb{N}_0$, for all $k$ in $\{1, 2,\ldots, n+1\}$:
\begin{equation}
\hultmannumber{n}{k}=\frac{1}{n+1}\sum_{i=1}^{n+1}[h^{k}]\fallingfactorial{(h+n-i+1)}{n+1}.\label{formulaH}%
\end{equation}
\end{theorem}
\begin{proof}
 By \Cref{unsigned-hultman-counts-factorisations}, $\hultmannumber{n}{k}$ counts the number of factorisations of a fixed $(n+1)$-cycle $\beta$ into the product $\rho\circ\omega$, with $c(\rho)=1$ and $c(\omega)=k$. This is clearly equivalent to enumerating factorisations of $\rho^{-1}$ into the product $\omega\circ\beta^{-1}$ under the same conditions; therefore, setting $\omega_{(n+1)}$ to $\beta^{-1}$ in \Cref{formula1}, we observe that $\hultmannumber{n}{k}$ is the coefficient of the monomial $h^{k}\ell$ in the polynomial $Q_{n+1}^{\mathbb{C}}(h,\ell)$, hence by \Cref{formula2} equals:
$$\hultmannumber{n}{k}=\frac{1}{n+1}\sum_{i=1}^{n+1}(-1)^{i-1}\frac
{[h^{k}]\fallingfactorial{(h+n-i+1)}{n+1}\times[\ell]\fallingfactorial{(\ell+n-i+1)}{n+1}}{(n-i+1)!(i-1)!}.
$$
Since for every $i$ in $\{1,2,\ldots,k+1\}$ we have
\begin{eqnarray*}
&& [\ell]\fallingfactorial{(\ell+n-i+1)}{n+1}\\
&=&[\ell](\ell+n-i+1)(\ell+n-i)\cdots(\ell+1)\ell(\ell-1)(\ell-2)\cdots(\ell-(i-1))\\
&=& (-1)^{i-1}(n-i+1)!(i-1)!,
\end{eqnarray*}
the above summation simplifies to the wanted expression, which completes the proof.
\end{proof}

Besides providing a new relation involving Hultman numbers, our
new formula will prove useful in obtaining simple proofs of known
results, as we will see in \Cref{sec:stats,sec:simpler-proofs}.
Moreover, we think that the interest of our formula also lies 
in the fact that the method used to prove it extends to the signed
case.

\section{An explicit formula for $\signedhultmannumber{n}{k}$}\label{sec:formula-signed-hultman}

We now turn our attention to the problem of computing \emph{signed} Hultman numbers, which we solve using ideas similar to those presented in the previous section.
The result is obtained by characterising the $2$-regular graphs that correspond to actual breakpoint graphs (\Cref{config-is-BG-iff-complement-hamiltonian} page~\pageref{config-is-BG-iff-complement-hamiltonian}), and then relating that characterisation to an enumeration result by \citet{hanlon-spectra}.

\subsection{Preliminaries}\label{sec:background-signed-hultman-numbers}

Following \citet{hanlon-spectra},
for some fixed $n$ in $\mathbb{N}_0$, let
$$
Q_{n}^{\mathbb{R}}(h,\ell)=\mathbb{E}(\mathrm{tr}((VV^{\mathrm{t}})^{n})),
$$
where $V$ is again a random $h\times \ell$ matrix, but this time with independent
standard \emph{real} normal entries.
\citet{hanlon-spectra} obtain two formulas for
$Q_{n}^{\mathbb{R}}(h,\ell)$.

Let $\mathcal{F}_{n}$ denote the set of
perfect matchings
of
$\{0,1,2,\ldots,2n-1\}$. In particular, let $\varepsilon\in
\mathcal{F}_{n}$ be the \textit{identity
perfect matching
} $\{\{i,n+i\}\
|\ 0\leq i\leq n-1\}$.
The first formula\footnote{See Corollary 3.6 of \citet{hanlon-spectra}.} for $Q_{n}^{\mathbb{R}}(h,\ell)$ is:
\begin{equation}
Q_{n}^{\mathbb{R}}(h,\ell)=\sum_{\delta\in\mathcal{F}_{n}}h^{c(\varepsilon\cup\delta)}%
\ell^{c(\delta\cup\delta_{(n)})},\label{1}
\end{equation}
where $\delta_{(n)}$ is a fixed
perfect matching
such that
$\varepsilon\cup\delta_{(n)}$ is
hamiltonian.

The second formula is based on partitions rather than on perfect matchings.

\begin{definition}\label{def:partition}
\cite{macdonald-symmetric} A (integer) \emph{partition}
$\lambda=(\lambda_1, \lambda_2, \ldots, \lambda_l)$ is a finite
sequence of integers called \emph{parts} such that
$\lambda_1\geq\lambda_2\geq\cdots\geq\lambda_l\geq 0$.
Its \emph{length}
is the number of non-zero parts it contains,
and if $\sum_{i=1}^{l}\lambda_i=n$, we call $\lambda$ a \emph{partition of} $n$, which we write as $\lambda\vdash n$.
\end{definition}

We consider any two partitions to be equivalent if we obtain the
same sequence when removing all parts that equal $0$. The notation
$\lambda=(1^{m_1}2^{m_2}\ldots$ $r^{m_r})$ is also frequently
used, and expresses the fact that exactly $m_i$ parts of $\lambda$
equal $i$. The reader must therefore bear in mind that when
working with partitions, the notation $a^b$ is more often to be
understood in the previous meaning, and not as ``$a$ to the power
$b$''.

The second formula\footnote{See Theorem 5.4 of \citet{hanlon-spectra}.} for $Q_{n}^{\mathbb{R}}(h,\ell)$ is:
\begin{equation}
Q_{n}^{\mathbb{R}}(h,\ell)=\sum_{\lambda}c_{\lambda}(2)F_{\lambda}(h)F_{\lambda}(\ell),\label{2}%
\end{equation}
where:
\begin{itemize}
 \item  $\lambda$ ranges over all partitions of $n$ of the form $(a,b,1^{n-a-b}%
)$, with either $a\geq b\geq 1$ or $a=n$ and $b=0$,

\item
 the function $F_{\lambda}:\mathbb{R} \to \mathbb{R}$
is defined as:
\begin{equation}\label{F_lambda}
F_{\lambda}(x)=2^{a-b}\fallingfactorial{(x/2+a-1)}{a-b}\fallingfactorial{(x+2b-2)}{n-a+b},
\end{equation}

\item
and the
coefficients $c_{\lambda}(2)$ are given as follows:
\begin{equation}\label{c_lambda1}
c_{\lambda}(2) =\frac{(-1)^{n+a-b+1}2^{a-b+1}n(2a-2b+1)(a-1)!}{\fallingfactorial{(n+a-b+1)}{2}%
\fallingfactorial{(n-a+b)}{2}(n-a-b)!(2a-1)!(b-1)!},
\end{equation}
if $\lambda=(a,b,1^{n-a-b})$, with $a\geq b\geq1$, and
\begin{equation}\label{c_lambda2}
c_{\lambda}(2) = \frac{2^{n}n!}{(2n)!},\mbox{ if } \lambda=(n).
\end{equation}

\end{itemize}

The numbers $c_{\lambda}(2)$ appear as coefficients in the
expansion of the $n^{\mbox{\tiny th}}$ power-sum function in terms of zonal
polynomials. For definitions and details, see for example
\citet{macdonald-symmetric}.

\subsection{Characterising valid breakpoint graphs}\label{sec:valid-breakpoint-graphs}

Recall that a breakpoint graph is a $2$-regular graph that is the
union of two
perfect matchings
of $\{0,1,\ldots, 2n+1\}$. We now make
the connection between signed Hultman numbers and the previously
mentioned results explicit.

\begin{definition}
 A \emph{configuration} is the union of two
perfect matchings
$\delta_B$ and $\delta_G$
of $\{0$, $1$, $\ldots$, $2n+1\}$, where $\delta_G=\{\{2i, 2i+1\}\ |\ 0\leq i\leq n\}$.
\end{definition}

Note that the above definition only slightly generalises \Cref{def:BG}, by allowing any choice of a perfect matching for $\delta_B$, whereas there are implicit constraints on the choice of $\delta_B$ in the definition of the breakpoint graph. By definition, every breakpoint graph is a configuration, but not every configuration is a breakpoint graph, as we will see below shortly. The following notion will help us characterise configurations that are breakpoint graphs.

\begin{definition}\label{def:complement}
The \emph{complement} of a
configuration
$C=\delta_B\cup\delta_G$,
denoted by
$\overline{C}=\delta_B\cup\overline{\delta_G}$,
is obtained
by replacing $\delta_G$ with $\overline{\delta_G}= \{\{2i-1,2i\}\
|\ 1\leq i\leq n\}\cup\{\{0,2n+1\}\}$.
\end{definition}

Before stating our characterisation of breakpoint graphs, we wish
to stress that \citet{elias-better-journal} previously used a
similar but different notion of complementation (they replace
$\delta_B$ with $\overline{\delta_B}$ -- whose definition we will omit here -- whereas we replace
$\delta_G$ with $\overline{\delta_G}$) to characterise valid
breakpoint graphs of \emph{unsigned} permutations. This is not
enough for our purpose, which is why we generalise their result
below to encompass \emph{signed} permutations as well.

\begin{lemma}\label{config-is-BG-iff-complement-hamiltonian}
A configuration $\delta_B\cup\delta_G$ is the breakpoint graph of
some signed permutation $\pi$ if and only if the complement
configuration $\delta_B\cup\overline{\delta_G}$ is hamiltonian.
\end{lemma}
\begin{proof}
We can easily see that the complement $\overline{BG(\pi)}$ of a
breakpoint graph is hamiltonian, since its edges are
$\{\{\pi'_{i}, \pi'_{i+1}\}\ |\ 0\leq i\leq 2n\}
\cup\{\{0,2n+1\}\}$.

Reciprocally, if the complement $\delta_B\cup\overline{\delta_G}$
of a configuration is hamiltonian, then we can recover the
elements of an unsigned permutation $\pi'=\langle 0\ \pi'_1\ \pi'_2\ \cdots$ $\pi'_{2n}\ 2n+1\rangle$ by
visiting the vertices along the hamiltonian cycle as follows: take
$0=\pi'_0$ as starting point, and follow the edge in $\delta_B$ that is incident to
$0$, setting the value of $\pi'_1$ to the other endpoint of that
edge. We then keep following the cycle, assigning the label of the
$i^{\mbox{\tiny th}}$ encountered vertex to $\pi'_i$ as we go,
ending with $2n+1=\pi'_{2n+1}$. Note that for every $0\leq i\leq
n$, the edge $\{\pi'_{2i+1},\pi'_{2i+2}\}$ belongs to
$\overline{\delta_G}$, and therefore we have
$|\pi'_{2i+1}-\pi'_{2i+2}|=1$. From the unsigned permutation
$\pi'$, we can therefore easily recover the corresponding signed
permutation $\pi$ in $S^\pm_n$, whose breakpoint graph is
$\delta_B\cup\delta_G$.
\end{proof}

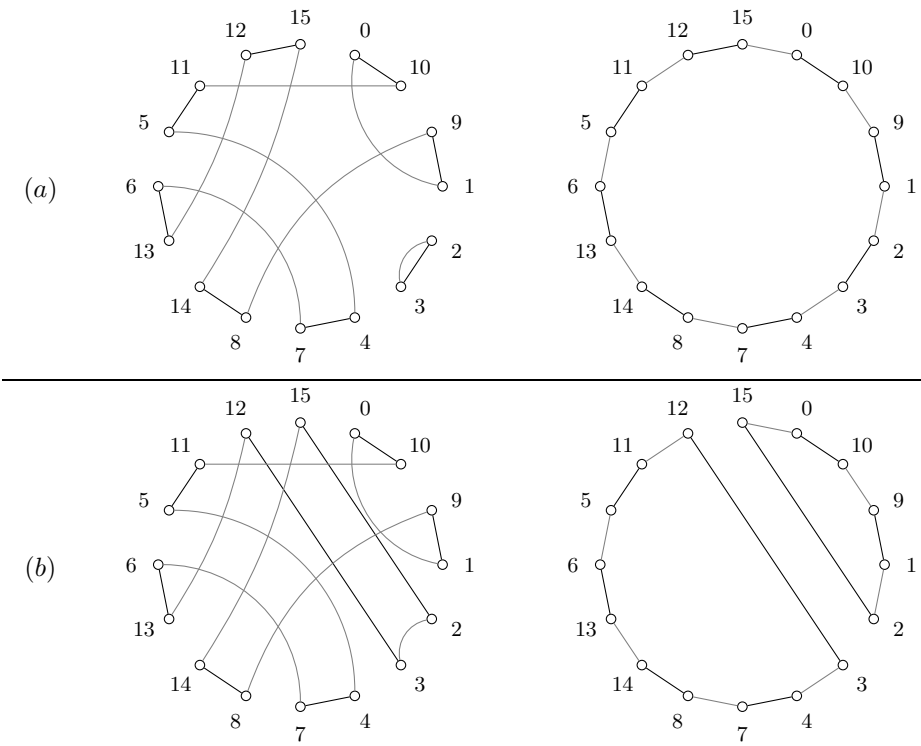
\begin{figure}[htbp]
\centering
\begin{tabular}{USS}
$(a)$ &
\begin{tikzpicture}[scale=.9]
  \foreach \position/\name in {3/0, 2/10, 1/9, 0/1, 15/2, 14/3, 13/4, 12/7, 11/8, 10/14, 9/13, 8/6, 7/5, 6/11, 5/12, 4/15}
     \path ( \position*22.5:2.5cm) node {\small $\name$};

  \foreach \position/\name in {3/v0, 2/v1, 1/v2, 0/v3, 15/v4, 14/v5, 13/v6, 12/v7, 11/v8, 10/v9, 9/v10, 8/v11, 7/v12, 6/v13, 5/v14, 4/v15} {
      \path (\position*22.5:2.1cm) node[draw,vertex] (\name) {}; 
  }
  \foreach \u/\v in {v0/v1, v2/v3, v4/v5, v6/v7, v8/v9, v10/v11, v12/v13, v14/v15}
      \draw (\u) -- (\v);
  \draw[gray] (v0) to [bend right=45] (v3);
  \draw[gray] (v4) to [bend right=45] (v5);
  \draw[gray] (v6) to [bend right=45] (v12);
  \draw[gray] (v7) to [bend right=45] (v11);
  \draw[gray] (v2) to [bend right=25] (v8);
  \draw[gray] (v1) -- (v13);
  \draw[gray] (v9) to [bend right=10] (v15);
  \draw[gray] (v10) to [bend right=10] (v14);
\end{tikzpicture}
&
\begin{tikzpicture}[scale=.9]
  \foreach \position/\name in {3/0, 2/10, 1/9, 0/1, 15/2, 14/3, 13/4, 12/7, 11/8, 10/14, 9/13, 8/6, 7/5, 6/11, 5/12, 4/15}
     \path ( \position*22.5:2.5cm) node {\small $\name$};

  \foreach \position/\name in {3/v0, 2/v1, 1/v2, 0/v3, 15/v4, 14/v5, 13/v6, 12/v7, 11/v8, 10/v9, 9/v10, 8/v11, 7/v12, 6/v13, 5/v14, 4/v15}
      \path (\position*22.5:2.1cm) node[draw,vertex] (\name) {}; 

  \foreach \u/\v in {v0/v1, v2/v3, v4/v5, v6/v7, v8/v9, v10/v11, v12/v13, v14/v15}
      \draw (\u) -- (\v);

  \foreach \u/\v in {v1/v2, v3/v4, v5/v6, v7/v8, v9/v10, v11/v12, v13/v14, v15/v0}
    \draw[gray] (\u) to 
                (\v);
\end{tikzpicture}
\\
\hline
$(b)$ &
\begin{tikzpicture}[scale=.9]
  \foreach \position/\name in {3/0, 2/10, 1/9, 0/1, 15/2, 14/3, 13/4, 12/7, 11/8, 10/14, 9/13, 8/6, 7/5, 6/11, 5/12, 4/15}
     \path ( \position*22.5:2.5cm) node {\small $\name$};

  \foreach \position/\name in {3/v0, 2/v1, 1/v2, 0/v3, 15/v4, 14/v5, 13/v6, 12/v7, 11/v8, 10/v9, 9/v10, 8/v11, 7/v12, 6/v13, 5/v14, 4/v15} {
      \path (\position*22.5:2.1cm) node[draw,vertex] (\name) {}; 
  }
  \foreach \u/\v in {v0/v1, v2/v3, v4/v15, v6/v7, v14/v5, v10/v11, v12/v13, v8/v9}
      \draw (\u) -- (\v);
  \draw[gray] (v0) to [bend right=45] (v3);
  \draw[gray] (v4) to [bend right=45] (v5);
  \draw[gray] (v6) to [bend right=45] (v12);
  \draw[gray] (v7) to [bend right=45] (v11);
  \draw[gray] (v2) to [bend right=25] (v8);
  \draw[gray] (v1) -- (v13);
  \draw[gray] (v9) to [bend right=10] (v15);
  \draw[gray] (v10) to [bend right=10] (v14);
\end{tikzpicture}
&
\begin{tikzpicture}[scale=.9]
  \foreach \position/\name in {3/0, 2/10, 1/9, 0/1, 15/2, 14/3, 13/4, 12/7, 11/8, 10/14, 9/13, 8/6, 7/5, 6/11, 5/12, 4/15}
     \path ( \position*22.5:2.5cm) node {\small $\name$};

  \foreach \position/\name in {3/v0, 2/v1, 1/v2, 0/v3, 15/v4, 14/v5, 13/v6, 12/v7, 11/v8, 10/v9, 9/v10, 8/v11, 7/v12, 6/v13, 5/v14, 4/v15}
      \path (\position*22.5:2.1cm) node[draw,vertex] (\name) {}; 

  \foreach \u/\v in {v0/v1, v2/v3, v4/v15, v6/v7, v14/v5, v10/v11, v12/v13, v8/v9}
      \draw (\u) -- (\v);

  \foreach \u/\v in {v1/v2, v3/v4, v5/v6, v7/v8, v9/v10, v11/v12, v13/v14, v15/v0}
    \draw[gray] (\u) to 
                (\v);

\end{tikzpicture}
\end{tabular}
\caption{$(a)$ The complement of the breakpoint graph from \Cref{fig:breakpoint-graph-example} is hamiltonian; $(b)$ a configuration whose complement is not hamiltonian.}
\label{fig:complement-breakpoint-graph}
\end{figure}

\Cref{fig:complement-breakpoint-graph}$(a)$ shows the complement of the breakpoint graph of \Cref{fig:breakpoint-graph-example} (page~\pageref{fig:breakpoint-graph-example}), which is hamiltonian. On the other hand, the complement of the configuration shown in \Cref{fig:complement-breakpoint-graph}$(b)$ is not hamiltonian.
We now show that \Cref{1} remains valid when replacing the identity
perfect matching
$\varepsilon$ with the
perfect matching
$\delta_{G}$ and choosing
$\overline{\delta_G}$ as the fixed
perfect matching
$\delta_{(n+1)}$, which clearly satisfies the condition that
$\delta_G\cup\overline{\delta_G}$ is hamiltonian as required. The proof can be
easily generalised to any choice of a
perfect matching
$\tau_{(n+1)}$ such
that $\delta_G\cup\tau_{(n+1)}$ is hamiltonian, but the following statement will be sufficient for our purposes.

\begin{lemma}\label{renaming-one-factors-is-fine}
For any $n$ in $\mathbb{N}_0$:
\begin{equation}
Q_{n+1}^{\mathbb{R}}(h,\ell)=\sum_{\tau\in\mathcal{F}_{n+1}}h^{c(\delta_G\cup\tau)}%
\ell^{c(\tau\cup\overline{\delta_G})}.
\end{equation}
\end{lemma}

\begin{proof} First, let us note that every
perfect matching
$\phi$ in $\mathcal{F}_{n+1}$
can be seen as a fixed-point-free involution, i.e.\ a permutation of $\{0,1,2,\ldots, 2n+1\}$ that decomposes into a collection of $2$-cycles only, by viewing each
edge of $\phi$ as a $2$-cycle. Therefore, conjugating $\phi$ by any permutation
of the same number of elements
is
a well-defined operation that simply renames the endpoints of the
given edges.
Let $\mu$ be the permutation defined by
$$
\mu:\{0,1,\ldots,2n+1\}\rightarrow\{0,1,\ldots,2n+1\}:i\mapsto\mu(i)=\left\{
\begin{array}{ll}
 i/2 & \mbox{if $i$ is even},\\
 \frac{i+2n+1}{2} & \mbox{otherwise}.
\end{array}
\right.
$$
As the example in \Cref{fig:mapping-using-mu} shows, $\delta_G$ can be mapped onto $\varepsilon=\mu\circ\delta_G\circ\mu^{-1}$, and we fix $\delta_{(n+1)}= \mu\circ\overline{\delta_G}\circ\mu^{-1}$.
Finally, observe that given any two
perfect matchings
$\phi_1$ and $\phi_2$ in $\mathcal{F}_{n+1}$, the graphs
$\mu\circ\phi_1\circ\mu^{-1}\cup\mu\circ\phi_2\circ\mu^{-1}$ and
$\phi_1\cup\phi_2$ are isomorphic, and hence
$c(\mu\circ\phi_1\circ\mu^{-1}\cup\mu\circ\phi_2\circ\mu^{-1})=c(\phi_1\cup\phi_2)$.
Taking $\delta = \mu \circ \tau \circ \mu^{-1}$, the following relations hold:
\begin{itemize}
 \item $c(\varepsilon\cup\delta)=c(\mu\circ\delta_G\circ\mu^{-1}\cup\mu\circ\tau\circ\mu^{-1})=c(\delta_G\cup\tau)$,
 \item $c(\delta\cup\delta_{(n+1)})=c(\mu\circ\tau\circ\mu^{-1}\cup\mu\circ\overline{\delta_G}\circ\mu^{-1})=c(\tau\cup\overline{\delta_G})$,
 \item
 $c(\varepsilon\cup\delta_{(n+1)})=c(\mu\circ\delta_G\circ\mu^{-1}\cup\mu\circ\overline{\delta_G}\circ\mu^{-1})=c(\delta_G\cup\overline{\delta_G})=1$,
\end{itemize}
and the formula in the statement follows from the above
relations, the bijectivity of conjugation,  and \Cref{1}.
\end{proof}

\begin{figure}[htbp]
 \centering
\begin{tikzpicture}
    \node at (-1,-0.5) {$\delta_G$};
    \node[vertex,draw] (0) at (0,0) [label=above:$0$] {} ;
    \node[vertex,draw] (1) at (0,-1) [label=below:$1$] {} ;
    \node[vertex,draw] (2) at (1,0) [label=above:$2$] {} ;
    \node[vertex,draw] (3) at (1,-1) [label=below:$3$] {} ;
    \node[vertex,draw] (4) at (2,0) [label=above:$4$] {} ;
    \node[vertex,draw] (5) at (2,-1) [label=below:$5$] {} ;
    \node[vertex,draw] (6) at (3,0) [label=above:$6$] {} ;
    \node[vertex,draw] (7) at (3,-1) [label=below:$7$] {} ;
    \node[vertex,draw] (8) at (4,0) [label=above:$8$] {} ;
    \node[vertex,draw] (9) at (4,-1) [label=below:$9$] {} ;
    \foreach \p/\q in {0/1, 2/3, 4/5, 6/7, 8/9}
        \draw (\p) -- (\q);
    \draw[->] (4.4, -0.5) -- (5.2, -0.5);
\begin{scope}[xshift=160pt]
    \node[vertex,draw] (0) at (0,0) [label=above:$0$] {} ;
    \node[vertex,draw] (1) at (0,-1) [label=below:$5$] {} ;
    \node[vertex,draw] (2) at (1,0) [label=above:$1$] {} ;
    \node[vertex,draw] (3) at (1,-1) [label=below:$6$] {} ;
    \node[vertex,draw] (4) at (2,0) [label=above:$2$] {} ;
    \node[vertex,draw] (5) at (2,-1) [label=below:$7$] {} ;
    \node[vertex,draw] (6) at (3,0) [label=above:$3$] {} ;
    \node[vertex,draw] (7) at (3,-1) [label=below:$8$] {} ;
    \node[vertex,draw] (8) at (4,0) [label=above:$4$] {} ;
    \node[vertex,draw] (9) at (4,-1) [label=below:$9$] {} ;
    \foreach \p/\q in {0/1, 2/3, 4/5, 6/7, 8/9}
        \draw (\p) -- (\q);
    \node at (5,-0.5) {$\varepsilon$};
\end{scope}
\begin{scope}[yshift=-70pt]
    \node at (-1,-0.5) {$\overline{\delta_G}$};
    \node[vertex,draw] (0) at (0,0) [label=above:$0$] {} ;
    \node[vertex,draw] (1) at (0,-1) [label=below:$1$] {} ;
    \node[vertex,draw] (2) at (1,0) [label=above:$2$] {} ;
    \node[vertex,draw] (3) at (1,-1) [label=below:$3$] {} ;
    \node[vertex,draw] (4) at (2,0) [label=above:$4$] {} ;
    \node[vertex,draw] (5) at (2,-1) [label=below:$5$] {} ;
    \node[vertex,draw] (6) at (3,0) [label=above:$6$] {} ;
    \node[vertex,draw] (7) at (3,-1) [label=below:$7$] {} ;
    \node[vertex,draw] (8) at (4,0) [label=above:$8$] {} ;
    \node[vertex,draw] (9) at (4,-1) [label=below:$9$] {} ;
    \foreach \p/\q in {0/9, 1/2, 3/4, 5/6, 7/8}
        \draw (\p) -- (\q);
    \draw[->] (4.4, -0.5) -- (5.2, -0.5);
\begin{scope}[xshift=160pt]
    \node[vertex,draw] (0) at (0,0) [label=above:$0$] {} ;
    \node[vertex,draw] (1) at (0,-1) [label=below:$5$] {} ;
    \node[vertex,draw] (2) at (1,0) [label=above:$1$] {} ;
    \node[vertex,draw] (3) at (1,-1) [label=below:$6$] {} ;
    \node[vertex,draw] (4) at (2,0) [label=above:$2$] {} ;
    \node[vertex,draw] (5) at (2,-1) [label=below:$7$] {} ;
    \node[vertex,draw] (6) at (3,0) [label=above:$3$] {} ;
    \node[vertex,draw] (7) at (3,-1) [label=below:$8$] {} ;
    \node[vertex,draw] (8) at (4,0) [label=above:$4$] {} ;
    \node[vertex,draw] (9) at (4,-1) [label=below:$9$] {} ;
    \foreach \p/\q in {0/9, 1/2, 3/4, 5/6, 7/8}
        \draw (\p) -- (\q);
    \node at (5,-0.5) {$\delta_{(n+1)}$};
\end{scope}
\end{scope}
\end{tikzpicture}
\caption{Mapping $\delta_G$ (resp.\ $\overline{\delta_G}$) onto $\varepsilon$ (resp.\ $\delta_{(n+1)}$) by conjugating them by $\mu=\left\langle 0\ 5\ 1\ 6\ 2\ 7\ 3\ 8\ 4\ 9\right\rangle$.}
\label{fig:mapping-using-mu}
\end{figure}
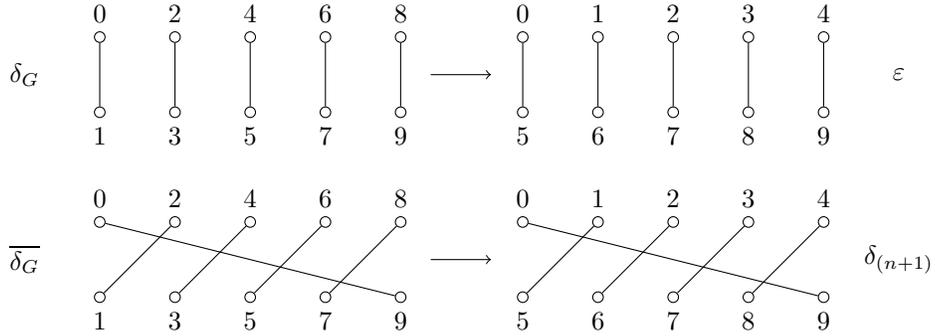

\subsection{Enumerating breakpoint graphs with $k$ cycles}

\Cref{config-is-BG-iff-complement-hamiltonian} implies that
enumerating signed permutations of $n$ elements whose breakpoint
graph decomposes into $k$ alternating cycles is equivalent to
enumerating
perfect matchings
$\tau$ in $\mathcal{F}_{n+1}$ verifying
$c(\delta_G \cup \tau)=k$ and $c(\tau\cup\overline{\delta_G})=1$,
where $\delta_G$ is defined in
\Cref{def:BG} page~\pageref{def:BG} and $\overline{\delta_G}$ is defined in
\Cref{def:complement} page~\pageref{def:complement}. Using
\Cref{renaming-one-factors-is-fine}, we thus obtain the following.

\begin{remark} For every
$k$ in $\{1,2,\ldots,n+1\}$, $\signedhultmannumber{n}{k}$ is the coefficient of the monomial
$h^{k}\ell$ in $Q_{n+1}^{\mathbb{R}}(h,$ $\ell)$.\label{rem}
\end{remark}

The second expression for $Q_{n+1}^{\mathbb{R}}(h,\ell)$ given in
\Cref{2} allows us to obtain the following explicit
formula for $\signedhultmannumber{n}{k}$.

\begin{theorem}\label{formula-for-signed-hultman-numbers} For all $n$ in $\mathbb{N}_0$, for all $k$ in $\{1, 2,\ldots, n+1\}$:
\begin{eqnarray}
\signedhultmannumber{n}{k}&=&\sum_{\lambda}c_{\lambda}(2)\times[h^{k}]F_{\lambda}(h)\nonumber\\
&&\times\frac{(-1)^{n-a-b}2^{a-b-1}(2b)!(a-1)!(n-a-b+2)!}{(2b-1)b!},
\end{eqnarray}
where $\lambda$ ranges over all partitions of $n+1$ of the form
$(a,b,1^{n-a-b+1})$, with $a\geq b\geq1$ or $a=n+1,b=0$, and where the
function $F_{\lambda}(\cdot)$ as well as the coefficients $c_{\lambda}(2)$ follow the definitions previously given in \Cref{sec:background-signed-hultman-numbers}\footnote{With the slight modification that $n$ needs to be replaced with $n+1$.}.
\end{theorem}
\begin{proof}
\Cref{rem} and \Cref{2} yield
\begin{equation}
\signedhultmannumber{n}{k}=\sum_{\lambda}c_{\lambda}(2)\times[h^{k}]F_{\lambda}(h)\times[\ell]F_{\lambda}(\ell),\label{form}
\end{equation}
where the sum over $\lambda$, the coefficients $c_{\lambda}(2)$
and the function $F_{\lambda}(\cdot)$ are as in the statement of
the present result. For a partition $\lambda$ of the form
$(a,b,1^{n-a-b+1})$, with $a\geq b\geq1$ or $a=n+1,b=0$, it is
easy to see that
\begin{equation}
[\ell]F_{\lambda}(\ell)=\frac{(-1)^{n-a-b}2^{a-b-1}(2b)!(a-1)!(n-a-b+2)!%
}{(2b-1)b!}.\label{coeff}
\end{equation}
Indeed:
\begin{enumerate}
 \item if $\lambda=(a,b,1^{n-a-b+1})$, with $a\geq
b\geq 1$, we have
\begin{eqnarray*}
F_{\lambda}(\ell)&= & 2^{a-b}(\ell/2+a-1)(\ell/2+a-2)\cdots(\ell/2+b)\\
& \times& (\ell+2b-2)(\ell+2b-3)\cdots(\ell+1)\\
&\times& \ell(\ell-1)\cdots(\ell-(n-a-b+2)).
\end{eqnarray*}

The coefficient of $\ell$ in the above expression equals
\begin{eqnarray*}
[\ell]F_{\lambda}(\ell) & =& 2^{a-b}\frac{(a-1)!}{(b-1)!}\times(2b-2)!(-1)^{n-a-b+2}(n-a-b+2)!\\
&=&\frac{(-1)^{n-a-b}2^{a-b-1}(2b)!(a-1)!(n-a-b+2)!}{(2b-1)b!}.
\end{eqnarray*}

\item if $\lambda=(n+1)$, i.e.\ $a=n+1$ and $b=0$, we have
\begin{eqnarray*}
F_{(n+1)}(\ell)&=&2^{n+1}\fallingfactorial{(\ell/2+n)}{n+1}\fallingfactorial{(\ell-2)}{0}\\
&=&2^{n+1}(\ell/2+n)(\ell/2+n-1)\cdots(\ell/2+1)\ell/2,
\end{eqnarray*}
so
$[\ell]F_{(n+1)}(\ell)=2^{n}n!$,
which
verifies \Cref{coeff}.
\end{enumerate}
The proof then follows from \Cref{form,coeff}.
\end{proof}

We conclude this section with \Cref{tab:signed-hultman-numbers}, which shows a few experimental values of the signed Hultman numbers. These values were previously obtained by the first author using the method described in a previous paper of hers~\cite{grusea-random}.

Note that for $k=1$, the sequence defined by $\signedhultmannumber{n}{1}$ for $n=1,2,\ldots$ corresponds to sequence A001171 in the On-Line Encyclopedia of Integer Sequences~\cite{sloane-oeis}.
As we will see in the next section, other known sequences also appear in that table.

\begin{table}
\centering \rotatebox{90} { \tiny
\begin{tabular}{|c|rrrrrrrrrrrr|}
\hline
 \backslashbox{\ \ \ \ \ $n$}{$k$} & $1$ & $2$ & $3$ & $4$ & $5$ & $6$ & $7$ & $8$  & $9$ & $10$ & $11$ & $12$\\
\hline
 1 &            1 &            1 &              &              &             &          &         &        &    &   & & \\
 2 &            4 &            3 &            1 &              &             &          &         &        &    &   & & \\
 3 &           20 &           21 &            6 &            1 &             &          &         &        &    &   & & \\
 4 &          148 &          160 &           65 &           10 &           1 &          &         &        &    &   & & \\
 5 &       1\,348 &       1\,620 &          701 &          155 &          15 &        1 &         &        &    &   & & \\
 6 &      15\,104 &      19\,068 &       9\,324 &       2\,247 &         315 &       21 &       1 &        &    &   & & \\
 7 &     198\,144 &     264\,420 &     138\,016 &      38\,029 &      5\,908 &      574 &      28 &      1 &    &   & & \\
 8 &  2\,998\,656 &  4\,166\,880 &  2\,325\,740 &     692\,088 &    124\,029 &  13\,524 &     966 &     36 &  1 &   & & \\
 9 & 51\,290\,496 & 74\,011\,488 & 43\,448\,940 & 13\,945\,700 & 2\,723\,469 & 344\,961 & 27\,930 & 1\,530 & 45 & 1 & & \\
10 & 979\,732\,224 & 1\,459\,381\,440 & 897\,020\,784 & 305\,142\,068 & 64\,711\,856 & 8\,996\,295 & 850\,905 & 53\,262 & 2\,310 & 55 & 1 & \\
11 & 20\,661\,458\,688 & 31\,674\,232\,128 & 20\,241\,273\,264 & 7\,255\,047\,116 & 1\,640\,552\,028 & 249\,029\,717 & 26\,004\,330 & 1\,910\,403 & 95\,304 & 3\,355 & 66 & 1 \\
\hline
\end{tabular}
\normalsize
}
\caption{A few values of
$\signedhultmannumber{n}{k}$} \label{tab:signed-hultman-numbers}
\end{table}

\section{Special cases}\label{sec:special-cases}

The expression obtained in \Cref{formula-for-signed-hultman-numbers} allows us to compute $\signedhultmannumber{n}{k}$ for all valid values of $n$ and $k$, but we must acknowledge that even though the formula is suited for practical use, it is unfortunately quite complicated and difficult to manipulate.
Simpler expressions do however exist for some particular cases, as we will show below. We will rely a lot on \Cref{config-is-BG-iff-complement-hamiltonian} in this section, and decide to use a slightly different layout for the breakpoint graph: labels are omitted for clarity, and grey edges rather than black edges are now laid out on a circle, so that computing the complement of a given configuration simply amounts to shifting grey edges sideways by one position. In order to make verifications easier for the reader, we also draw edges in the complement as dotted edges. The following particular cases are easy to verify:

\begin{enumerate}
 \item $\signedhultmannumber{n}{k}=0$ for all $k<1$ and all $k>n+1$ (trivial);
 \item $\signedhultmannumber{n}{n+1}=1$, since the only permutation whose breakpoint graph decomposes into $n+1$ cycles is $\iota$;
 \item $\signedhultmannumber{n}{n} = \binom{n+1}{2}$, since enumerating such permutations comes down to counting breakpoint graphs whose cycles all have length $1$, except for one that has length $2$. This in turn is equivalent to enumerating the ways in which one can connect any two of the $n+1$
grey
edges by black edges so as to obtain a valid configuration (with respect to \Cref{config-is-BG-iff-complement-hamiltonian}); as can be verified on \Cref{fig:all-single-two-cycles}, only one of the two possible choices of black edges (namely, configuration $(b)$) is valid, and the equality follows from the fact that there are $\binom{n+1}{2}$ possible ways to select two grey edges out of $n+1$.

\begin{figure}[htbp]
\centering
 \begin{tikzpicture}[scale=.5]
  \path ( 35:2.1cm) node[draw,vertex] (a1) {}; 
  \path ( 55:2.1cm) node[draw,vertex] (a2) {}; 
  \path (125:2.1cm) node[draw,vertex] (b1) {}; 
  \path (145:2.1cm) node[draw,vertex] (b2) {}; 
  \path (215:2.1cm) node[draw,vertex] (c1) {}; 
  \path (235:2.1cm) node[draw,vertex] (c2) {}; 
  \path (305:2.1cm) node[draw,vertex] (d1) {}; 
  \path (325:2.1cm) node[draw,vertex] (d2) {}; 

  \draw[gray] (a1) -- (a2);
  \draw[gray] (b1) -- (b2);
  \draw[gray] (c1) -- (c2);
  \draw[gray] (d1) -- (d2);

  \path (-10:2.1cm) node[vertex] (W1) {}; 
  \path ( 10:2.1cm) node[vertex] (W2) {}; 

  \path ( 80:2.1cm) node[draw,vertex] (X1) {}; 
  \path (100:2.1cm) node[draw,vertex] (X2) {}; 
  \path (170:2.1cm) node[vertex] (Y1) {}; 
  \path (190:2.1cm) node[vertex] (Y2) {}; 

  \path (260:2.1cm) node[draw,vertex] (Z1) {}; 
  \path (280:2.1cm) node[draw,vertex] (Z2) {}; 

  \draw[dotted] (W1) -- (W2);
  \draw[gray] (X1) -- (X2);
  \draw[dotted] (Y1) -- (Y2);
  \draw[gray] (Z1) -- (Z2);
  \draw (a1) to [bend left=45] (a2);
  \draw (b1) to [bend left=45] (b2);
  \draw (d1) to [bend left=45] (d2);
  \draw (c1) to [bend left=45] (c2);

\draw[dotted] (a2) -- (X1);
\draw[dotted] (X2) -- (b1);
\draw[dotted] (b2) -- (Y1);
\draw[dotted] (Y2) -- (c1);
\draw[dotted] (c2) -- (Z1);
\draw[dotted] (Z2) -- (d1);
\draw[dotted] (d2) -- (W1);
\draw[dotted] (W2) -- (a1);

\draw (X1) -- (Z2);
\draw (X2) -- (Z1);
    \node at (0,-2.5) {$(a)$};
\begin{scope}[xshift=185pt]
  \path ( 35:2.1cm) node[draw,vertex] (a1) {}; 
  \path ( 55:2.1cm) node[draw,vertex] (a2) {}; 
  \path (125:2.1cm) node[draw,vertex] (b1) {}; 
  \path (145:2.1cm) node[draw,vertex] (b2) {}; 
  \path (215:2.1cm) node[draw,vertex] (c1) {}; 
  \path (235:2.1cm) node[draw,vertex] (c2) {}; 
  \path (305:2.1cm) node[draw,vertex] (d1) {}; 
  \path (325:2.1cm) node[draw,vertex] (d2) {}; 

  \draw[gray] (a1) -- (a2);
  \draw[gray] (b1) -- (b2);
  \draw[gray] (c1) -- (c2);
  \draw[gray] (d1) -- (d2);

  \path (-10:2.1cm) node[vertex] (W1) {}; 
  \path ( 10:2.1cm) node[vertex] (W2) {}; 

  \path ( 80:2.1cm) node[draw,vertex] (X1) {}; 
  \path (100:2.1cm) node[draw,vertex] (X2) {}; 
  \path (170:2.1cm) node[vertex] (Y1) {}; 
  \path (190:2.1cm) node[vertex] (Y2) {}; 

  \path (260:2.1cm) node[draw,vertex] (Z1) {}; 
  \path (280:2.1cm) node[draw,vertex] (Z2) {}; 

  \draw[dotted] (W1) -- (W2);
  \draw[gray] (X1) -- (X2);
  \draw[dotted] (Y1) -- (Y2);
  \draw[gray] (Z1) -- (Z2);
  \draw (a1) to [bend left=45] (a2);
  \draw (b1) to [bend left=45] (b2);
  \draw (d1) to [bend left=45] (d2);
  \draw (c1) to [bend left=45] (c2);

\draw[dotted] (a2) -- (X1);
\draw[dotted] (X2) -- (b1);
\draw[dotted] (b2) -- (Y1);
\draw[dotted] (Y2) -- (c1);
\draw[dotted] (c2) -- (Z1);
\draw[dotted] (Z2) -- (d1);
\draw[dotted] (d2) -- (W1);
\draw[dotted] (W2) -- (a1);

\draw (X1) -- (Z1);
\draw (X2) -- (Z2);

    \node at (0,-2.5) {$(b)$};
\end{scope}
\end{tikzpicture}
\caption{The two forms of $2$-cycles that may arise in a breakpoint graph. Only four $1$-cycles are shown in each graph, but there can be any number of them.}
\label{fig:all-single-two-cycles}
\end{figure}
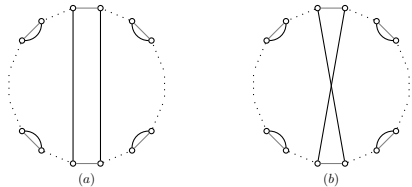
\end{enumerate}

We now show how one can obtain a simple and explicit formula for $\signedhultmannumber{n}{n-1}$. Although the formula is quite simple, we hope that the proof will convince the reader of the shortcomings of a case analysis in this setting.

\begin{proposition}
For all $n\geq 1$, we have $\signedhultmannumber{n}{n-1}=5\binom{n+1}{4}+4\binom{n+1}{3}$.
\end{proposition}
\begin{proof}
Note that $\signedhultmannumber{n}{n-1}$ is the number of permutations whose breakpoint graph contains either one $3$-cycle or
two $2$-cycles,
all other cycles having length $1$ in both cases:
\begin{enumerate}
\item the number of permutations satisfying the first condition is the number of ways to connect three
grey
edges in the breakpoint graph in such a way that the complement configuration is hamiltonian (see \Cref{config-is-BG-iff-complement-hamiltonian}). As \Cref{fig:all-three-cycles} shows, there are eight possible ways to create such a configuration, only four of which are valid (namely, configurations $(a)$, $(b)$, $(c)$ and $(d)$). The reader can easily verify that the other configurations are invalid by replacing
grey
edges with dotted edges.

\begin{figure}[htbp]
\centering
\begin{tikzpicture}[scale=.5]
\basisthreecycle
\draw (a2) to [bend right=30] (c1);
\draw (c2) to [bend right=30] (b2);
\draw (a1) to [bend left=30] (b1);

    \node at (0,-2.5) {$(a)$};
\begin{scope}[xshift=185pt]
\basisthreecycle
\draw (a1) to [bend right=30] (c1);
\draw (a2) to [bend left=30] (b2);
\draw (b1) to [bend left=30] (c2);

    \node at (0,-2.5) {$(b)$};

\end{scope}

\begin{scope}[xshift=375pt]

\basisthreecycle
\draw (a2) to [bend right=30] (c2);
\draw (a1) to [bend left=30] (b2);
\draw (b1) to [bend left=30] (c1);

    \node at (0,-2.5) {$(c)$};

\end{scope}
\begin{scope}[xshift=565pt]
\basisthreecycle
\draw (a2) to [bend right=30] (c1);
\draw (a1) to [bend left=30] (b2);
\draw (b1) to [bend left=30] (c2);

    \node at (0,-2.5) {$(d)$};

\end{scope}


\begin{scope}[yshift=-150pt]
\basisthreecycle

\draw (a1) to [bend right=45] (c1);
\draw (c2) to [bend right=45] (b2);
\draw (b1) to [bend right=45] (a2);

    \node at (0,-2.5) {$(e)$};

\begin{scope}[xshift=185pt]
\basisthreecycle
\draw (a2) to [bend right=45] (c2);
\draw (c1) to [bend right=45] (b2);
\draw (b1) to [bend right=45] (a1);

    \node at (0,-2.5) {$(f)$};

\end{scope}
\begin{scope}[xshift=375pt]
\basisthreecycle

\draw (a1) to [bend right=45] (c2);
\draw (c1) to [bend right=45] (b1);
\draw (b2) to [bend right=45] (a2);

    \node at (0,-2.5) {$(g)$};
\end{scope}

\begin{scope}[xshift=565pt]
\basisthreecycle

\draw (a1) to [bend right=45] (c2);
\draw (c1) to [bend right=45] (b2);
\draw (b1) to [bend right=45] (a2);

    \node at (0,-2.5) {$(h)$};

\end{scope}
\end{scope}

\end{tikzpicture}
\caption{All possible forms of $3$-cycles that may arise in a breakpoint graph. Only three $1$-cycles are shown in each graph, but there can be any number of them.}
\label{fig:all-three-cycles}
\end{figure}
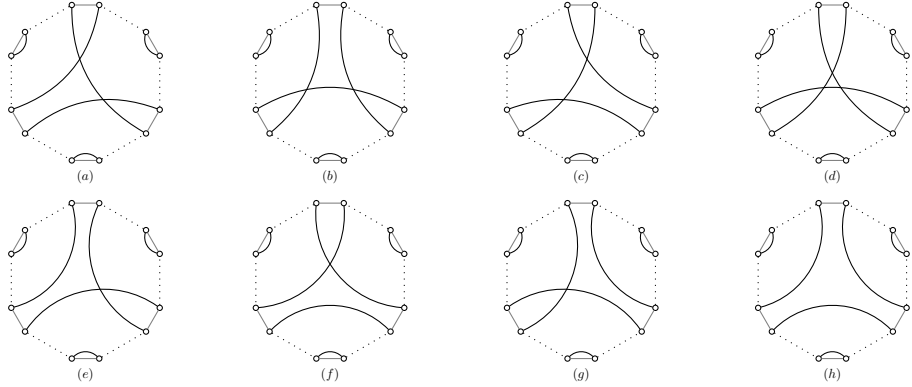

We obtain the rightmost term in the wanted expression by noting that only four of the eight possible $3$-cycles are valid, and there are $\binom{n+1}{3}$ ways to select three
grey
edges out of $n+1$.

\item the number of permutations satisfying the second condition can be constructed by choosing four
grey
edges, then connecting them by pairs while ensuring that the resulting configuration is valid.
\Cref{fig:all-two-cycles} shows all possible configurations with two cycles of length two.

\begin{figure}[htbp]
\centering
\begin{tikzpicture}[scale=.5]
\basispairtwocycles
\draw (a2) to [bend left=45] (b1);
\draw (a1) to [bend left=45] (b2);
\draw (c2) to [bend left=45] (d1);
\draw (d2) to [bend right=45] (c1);

    \node at (0,-2.5) {$(a)$};
\begin{scope}[xshift=185pt]
\basispairtwocycles

\draw (a2) to [bend left=45] (b2);
\draw (a1) to [bend left=45] (b1);
\draw (c2) to [bend left=45] (d2);
\draw (d1) to [bend right=45] (c1);
    \node at (0,-2.5) {$(b)$};

\end{scope}
\begin{scope}[xshift=375pt]
\basispairtwocycles
\draw (a2) to [bend left=45] (b2);
\draw (a1) to [bend left=45] (b1);
\draw (c1) to [bend left=45] (d2);
\draw (c2) to [bend left=45] (d1);
    \node at (0,-2.5) {$(c)$};

\end{scope}
\begin{scope}[xshift=565pt]
\basispairtwocycles
\draw (a2) to [bend left=45] (b1);
\draw (a1) to [bend left=45] (b2);
\draw (c2) to [bend left=45] (d2);
\draw (d1) to [bend right=45] (c1);
    \node at (0,-2.5) {$(d)$};

\end{scope}
\begin{scope}[yshift=-150pt]
\basispairtwocycles
\draw (a2) to [bend right=45] (d1);
\draw (a1) to [bend right=45] (d2);
\draw (b2) to [bend left=45] (c1);
\draw (b1) to [bend left=45] (c2);

    \node at (0,-2.5) {$(e)$};

\begin{scope}[xshift=185pt]
\basispairtwocycles
\draw (a2) to [bend right=45] (d2);
\draw (a1) to [bend right=45] (d1);
\draw (b2) to [bend left=45] (c2);
\draw (b1) to [bend left=45] (c1);

    \node at (0,-2.5) {$(f)$};

\end{scope}
\begin{scope}[xshift=375pt]
\basispairtwocycles
\draw (a2) to [bend right=45] (d2);
\draw (a1) to [bend right=45] (d1);
\draw (b2) to [bend left=45] (c1);
\draw (b1) to [bend left=45] (c2);

    \node at (0,-2.5) {$(g)$};

\end{scope}
\begin{scope}[xshift=565pt]
\basispairtwocycles
\draw (a2) to [bend right=45] (d1);
\draw (a1) to [bend right=45] (d2);
\draw (b2) to [bend left=45] (c2);
\draw (b1) to [bend left=45] (c1);
    \node at (0,-2.5) {$(h)$};

\end{scope}
\end{scope}

\begin{scope}[yshift=-300pt]
\basispairtwocycles
\draw (a2) to (c1);
\draw (a1) to (c2);
\draw (b2) to (d1);
\draw (b1) to (d2);

    \node at (0,-2.5) {$(i)$};

\begin{scope}[xshift=185pt]
\basispairtwocycles
\draw (a2) to (c2);
\draw (a1) to (c1);
\draw (b2) to (d2);
\draw (b1) to (d1);

    \node at (0,-2.5) {$(j)$};

\end{scope}
\begin{scope}[xshift=375pt]
\basispairtwocycles
\draw (a2) to (c1);
\draw (a1) to (c2);
\draw (b2) to (d2);
\draw (b1) to (d1);

    \node at (0,-2.5) {$(k)$};

\end{scope}
\begin{scope}[xshift=565pt]
\basispairtwocycles
\draw (a2) to (c2);
\draw (a1) to (c1);
\draw (b2) to (d1);
\draw (b1) to (d2);
    \node at (0,-2.5) {$(l)$};

\end{scope}
\end{scope}
\end{tikzpicture}
\caption{All possible pairs of $2$-cycles that may arise in a breakpoint graph. Only four $1$-cycles are shown in each graph, but there can be any number of them.}
\label{fig:all-two-cycles}
\end{figure}
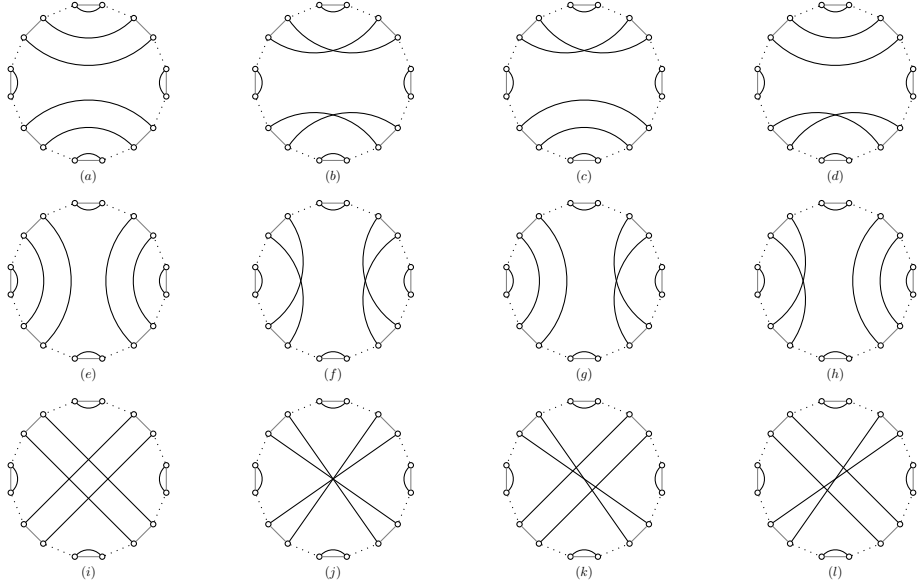

The reader can again easily verify the validity of all configurations by replacing
grey
edges with dotted edges.
Only five possible configurations with two $2$-cycles are valid (namely, configurations $(b)$, $(f)$, $(i)$, $(k)$ and $(l)$) out of the twelve shown in \Cref{fig:all-two-cycles}, and there are $\binom{n+1}{4}$ ways to select two pairs of
grey
edges out of $n+1$, which yields the leftmost term in the wanted expression and completes the proof.
\end{enumerate}
\end{proof}

\section{Simpler proofs of previous results}\label{sec:simpler-proofs}

\Cref{new-formula-for-hultman-numbers} allows us to obtain a new
proof of \citeauthor{bona-average}'s formula
(\Cref{eqn:hultman-numbers-bona-flynn} page~\pageref{eqn:hultman-numbers-bona-flynn}).

\begin{corollary}\label{corollary:new-proof-of-bonas-result}
\cite{bona-average} For all $n$ in $\mathbb{N}_0$:
\begin{equation*}
\hultmannumber{n}{k}=\left\{
\begin{array}{ll}
\firststirlingnumber{n+2}{k}/\binom{n+2}{2} & \mbox{if $n-k$ is odd,}\\
0 & \mbox{otherwise.}
\end{array}
\right.
\end{equation*}
\end{corollary}
\begin{proof}
The key idea of the proof is the fact that, for every
$i=1,2,\ldots,n+1$, we have
\begin{equation}
\fallingfactorial{(h+n-i+1)}{n+1}=\frac{1}{n+2}\left(\risingfactorial{(h-i+1)}{n+2}-\risingfactorial{(h-i)}{n+2}\right),\label{idea}
\end{equation}
since
\begin{eqnarray*}
&&\frac{1}{n+2}\left(\risingfactorial{(h-i+1)}{n+2}-\risingfactorial{(h-i)}{n+2}\right)\\
&=&\frac{1}{n+2}\left((h-i+1)\cdots(h+n-i+2)-(h-i)\cdots(h+n-i+1)\right)\\
&=&\frac{1}{n+2}(h-i+1)\cdots(h+n-i+1)\left((h+n-i+2)-(h-i)\right)\\
&=&\fallingfactorial{(h+n-i+1)}{n+1}.
\end{eqnarray*}

\noindent Summing over $i$ in \Cref{idea}, we obtain:
\begin{eqnarray*}
& &\frac{1}{n+1}\sum_{i=1}^{n+1}\fallingfactorial{(h+n-i+1)}{n+1}\\
& =&\frac{1}{(n+1)(n+2)}\sum_{i=1}^{n+1}\left(\risingfactorial{(h-i+1)}{n+2}-\risingfactorial{(h-i)}{n+2}\right)\\
& =&\frac{1}{(n+1)(n+2)}\left(\risingfactorial{h}{n+2}-\risingfactorial{(h-n-1)}{n+2}\right)\\
& =&\frac{1}{(n+1)(n+2)}\left(\risingfactorial{h}{n+2}-\fallingfactorial{h}{n+2}\right).
\end{eqnarray*}

By \Cref{fallingfactorial,risingfactorial}, the coefficient of $h^k$ in $\risingfactorial{h}{n+2}$ is
$\firststirlingnumber{n+2}{k}$ and the
coefficient of $h^k$ in $\fallingfactorial{h}{n+2}$ is
$(-1)^{n-k}\firststirlingnumber{n+2}{k}$.
Using \Cref{formulaH}, we conclude that
$$\hultmannumber{n}{k}=\left\{
\begin{array}{ll}
\frac{2}{(n+1)(n+2)}\firststirlingnumber{n+2}{k} & \mbox{if $n-k$ is odd,}\\
0 & \mbox{otherwise,}
\end{array}
\right.
$$
which completes the proof.
\end{proof}

\Cref{new-formula-for-hultman-numbers} also allows us to obtain a
simple proof of a binomial identity previously obtained by
\citet{Sury}.

\begin{corollary}\label{prop:new-proof-of-surys-result}
\cite{Sury} For all $n$ in $\mathbb{N}_0$:
$$\sum_{i=0}^{n}\frac{(-1)^{i}}{\binom{n}{i}}=(1+(-1)^n)\frac{n+1}{n+2}.$$
\end{corollary}
\begin{proof}
Setting $k$ to $1$ in \Cref{formulaH} (page~\pageref{formulaH})
yields
$$\hultmannumber{n}{1} =\frac{1}{n+1}\sum_{i=1}^{n+1}(-1)^{i-1}(n-i+1)!(i-1)! =\frac{n!}{n+1}\sum_{i=0}^{n}\frac{(-1)^{i}}{\binom{n}{i}}.$$
On the other hand, as previously observed\footnote{The result can
also be easily derived from
\Cref{eqn:hultman-numbers-bona-flynn}.} by
\citet{doignon-hultman}, we have:
$$
\hultmannumber{n}{1}=\left\{
\begin{array}{ll}
\frac{2n!}{n+2} & \mbox{if $n$ is even,}\\
0 & \mbox{otherwise,}
\end{array}
\right.
$$
which completes the proof.
\end{proof}


\section{Expected value and variance of the Hultman numbers}\label{sec:stats}

In order to gain more insight into the distribution of
the Hultman numbers,
we will now investigate the
question of computing the expected value and variance of the
number of cycles in breakpoint graphs, both for unsigned and for signed permutations.

It will also be interesting to see how these values compare to the
expected value and variance of the number of cycles in the usual
disjoint cycle decomposition of a uniform random unsigned
permutation $\pi$ in $S_n$. We recall here (see e.g.\
\citet{wilf-generatingfunctionology}) the exact values of
these quantities:
\begin{eqnarray*}
\mathbb{E}(c(\pi)) &=& H_n ,\\
\mathrm{Var}(c(\pi)) &=& H_n - \sum_{k=1}^n \frac{1}{k^2},
\end{eqnarray*}
as well as their asymptotic behaviour when $n \to \infty$:
\begin{eqnarray}
\mathbb{E}(c(\pi)) &=& \log(n) + \gamma + o(1)\label{eqn:mean-cycles-asymptotic} ,\\
\mathrm{Var}(c(\pi)) &=& \log(n) + \gamma - \frac{\pi^2}{6} +
o(1)\label{eqn:var-cycles-asymptotic},
\end{eqnarray}
where $H_n$ denotes the
$n^{\mbox{\tiny th}}$ harmonic number $H_n=\sum_{i=1}^n
\frac{1}{i}$ and $\gamma$ denotes the Euler-Mascheroni constant.
As usual, $o(1)$ denotes a quantity that converges to $0$ as $n
\to \infty$.

\subsection{The unsigned case}

\citet{bona-average} already proved a formula for computing the
expected number of cycles in the breakpoint graph of a uniform
random unsigned permutation. In this section we provide a new
proof of their result and also give an explicit formula for the
variance of this distribution.
We start by computing the generating function of the Hultman
numbers.

\begin{lemma}\label{generating-function-for-hultman-numbers}
For all $n\in\mathbb{N}_0$, we have:
$$
F(x) =
\sum_{k=0}^{n+1}\hultmannumber{n}{k} x^k =
\frac{\risingfactorial{x}{n+2}-\fallingfactorial{x}{n+2}}{2\binom{n+2}{2}}.
$$
\end{lemma}
\begin{proof}
The derivation is straightforward:
\begin{eqnarray*}
\sum_{k=0}^{n+1}\hultmannumber{n}{k} x^k &=& \frac{1}{\binom{n+2}{2}}\sum_{k=0}^{n+1}\frac{\firststirlingnumber{n+2}{k}-(-1)^{n+2-k}\firststirlingnumber{n+2}{k}}{2}x^k  \quad\mbox{(by  \Cref{eqn:hultman-numbers-bona-flynn})}\\
 &=& \frac{1}{2\binom{n+2}{2}}\left(\sum_{k=0}^{n+2}\firststirlingnumber{n+2}{k}x^k-\sum_{k=0}^{n+2}(-1)^{n+2-k}\firststirlingnumber{n+2}{k}x^k\right) \\
 &=& \frac{\risingfactorial{x}{n+2}-\fallingfactorial{x}{n+2}}{2\binom{n+2}{2}}. \quad\quad\quad\quad\quad\quad\quad\mbox{(by \Cref{risingfactorial,fallingfactorial}})
\end{eqnarray*}
\end{proof}

Knowing the generating function allows us to easily derive the
expected value and the variance of the number of cycles in the
breakpoint graph of a uniform random unsigned permutation. For
this purpose, we first need to compute some derivatives of the
generating function.

\begin{lemma}\label{derivatives-F}
For all $n\in\mathbb{N}_0$, we have:
\begin{eqnarray*}
F(1) &= &n!,\\
F'(1) &=&\frac{1}{2\binom{n+2}{2}}\left\{(n+2)!H_{n+2}+(-1)^{n-1}
n!\right\},\\
F''(1)&=&\frac{1}{2\binom{n+2}{2}}\left\{(n+2)!\left(H_{n+2}^2-\sum_{k=1}^{n+2}\frac{1}{k^2}\right)+2(-1)^n
n!(H_n-1)\right\}.
\end{eqnarray*}
\end{lemma}

\begin{proof}
We obtain the three expressions separately.
\begin{enumerate}
 \item
For the first expression, note that, by definition,
$F(1)=\sum_{k=1}^{n+1} \hultmannumber{n}{k}$, which is simply the
total number of permutations of $n$ elements and therefore equals
$n!$.

\item We simplify the computation of $F'(x)$ by writing
$\fallingfactorial{x}{n+2}=(x-1)g(x)$, with
$$g(x)=x \prod_{i=2}^{n+1} (x-i).$$
With this notation we have
$$
F(x) = \frac{\risingfactorial{x}{n+2}-(x-1)g(x)}{2\binom{n+2}{2}}.
$$
We thus obtain
$$F'(x)=\frac{1}{2\binom{n+2}{2}}\left(\risingfactorial{x}{n+2}\sum_{i=0}^{n+1}\frac{1}{x+i}-g(x)-(x-1)g'(x)\right).$$
At $x=1$ we have $\risingfactorial{1}{n+2}=(n+2)!$ and
$g(1)=(-1)^n n!$, and hence the stated formula for $F'(1)$
follows.

\item Finally, the second derivative of $F$ is given by
\begin{equation*}
F''(x)=\frac{1}{2\binom{n+2}{2}}\left(\risingfactorial{x}{n+2}\!
\!\!\!\!\! \sum_{0\leq i \neq j \leq n+1} \! \frac{1}{(x+i)(x+j)}
-2g'(x)- (x-1)g''(x)\right).
\end{equation*}
The above sum evaluated at $x=1$ equals
\begin{align*}
\sum_{0\leq i \neq j \leq n+1} \frac{1}{(1+i)(1+j)} &= \sum_{i, j
= 0}^{n+1} \frac{1}{(1+i)(1+j)} - \sum_{i = 0}^{n+1} \frac{1}{(1+i)^2}\\
&= \left(\sum_{i=0}^{n+1} \frac{1}{1+i}\right)^2 - \sum_{i = 0}^{n+1} \frac{1}{(1+i)^2}\\
& = H_{n+2}^2-\sum_{k=1}^{n+2}\frac{1}{k^2}.
\end{align*}
We also have
$$g'(x) = g(x)\left(\frac{1}{x}+ \sum_{i=2}^{n+1} \frac{1}{x-i} \right),$$
and thus
$$
g'(1) = g(1)\left(1- \sum_{i=2}^{n+1} \frac{1}{i-1} \right)=
(-1)^n n! (1-H_n).
$$
Using these expressions in the formula for $F''(x)$ above,
evaluated at $x=1$, gives the formula in the statement.
\end{enumerate}
\end{proof}

The recovery of the expected value of the Hultman numbers,
previously obtained by
\citet{bona-average},
is now an easy
task.

\begin{theorem}\label{mean-of-unsigned-hultman-numbers}
 \cite{bona-average} For all $n\in\mathbb{N}_0$, the expected number of cycles in the breakpoint graph of a uniform random unsigned
 permutation $\pi$ of $n$ elements is
 $$\mathbb{E}(c(BG(\pi)))= H_n + \frac{1}{\left\lfloor (n+2)/2\right\rfloor}.$$
\end{theorem}
\begin{proof}
As is well-known (see e.g.\ \citet{wilf-generatingfunctionology}),
the expected value
can be
obtained from the generating function
$F(x)$
by
the formula $F'(1)/F(1)$.
Using the formulas for $F(1)$ and $F'(1)$
obtained
in
\Cref{derivatives-F}, we obtain that the expected value of the
Hultman numbers equals
$$
 \frac{F'(1)}{F(1)} = H_{n+2}+\frac{(-1)^{n-1}}{(n+1)(n+2)},
$$
which is easily seen to be equivalent to the expression in the
statement.
\end{proof}

Furthermore,
knowing
the generating function
also
allows us to
compute the variance of the Hultman numbers. We prove the
following result.

\begin{theorem}\label{variance-of-unsigned-hultman-numbers}
For all $n\in\mathbb{N}_0$, the variance of the number of cycles
in the breakpoint graph of a uniform random unsigned permutation
$\pi$ of $n$ elements is
$$\mathrm{Var}(c(BG(\pi))) = H_{n+2}  -
\sum_{k=1}^{n+2}\frac{1}{k^2} +
\frac{(-1)^n(2H_{n+2}+2H_n-3)}{(n+1)(n+2)}-
\frac{1}{((n+1)(n+2))^2}.$$
\end{theorem}
\begin{proof}
The variance
can be obtained from the generating function $F(x)$ by
the following formula (see e.g.\
\citet{wilf-generatingfunctionology}):
$$(\log F)'(1)+(\log F)''(1) = \frac{F'(1)}{F(1)} + \frac{F''(1)}{F(1)}-
 \left(\frac{F'(1)}{F(1)}\right)^2.$$

Using the formulas for $F(1)$, $F'(1)$ and $F''(1)$ obtained in
\Cref{derivatives-F}, we obtain that the variance of the Hultman
numbers equals
\begin{eqnarray*}
&&\frac{F'(1)}{F(1)} + \frac{F''(1)}{F(1)} -
 \left(\frac{F'(1)}{F(1)}\right)^2\\
&=& H_{n+2} + \frac{(-1)^{n-1}}{(n+1)(n+2)} + H_{n+2}^2 -
\sum_{k=1}^{n+2}\frac{1}{k^2} +
\frac{2(-1)^n(H_n-1)}{(n+1)(n+2)}\\
&& - \left(H_{n+2} + \frac{(-1)^{n-1}}{(n+1)(n+2)}\right)^2\\
&=& H_{n+2}  - \sum_{k=1}^{n+2}\frac{1}{k^2} +
\frac{(-1)^n(2H_{n+2}+2H_n-3)}{(n+1)(n+2)}-
\frac{1}{((n+1)(n+2))^2}.
\end{eqnarray*}
\end{proof}

It is interesting to see how the mean and variance behave for
large $n$.
\begin{remark} The expected value and variance of the number of cycles in the breakpoint
graph of a uniform random unsigned permutation $\pi$ in $S_n$ have the
following asymptotical behaviour when $n \to \infty$:
\begin{eqnarray*}
\mathbb{E}(c(BG(\pi))) &=& \log(n) + \gamma + o(1) ,\\
\mathrm{Var}(c(BG(\pi))) & =& \log(n) + \gamma - \frac{\pi^2}{6} +
o(1).
\end{eqnarray*}
\end{remark}

\begin{proof}
For the expected value, the result simply follows from the fact
that $\mathbb{E}(c(BG(\pi)))= H_n + o(1)$ and $H_n = \log(n) +
\gamma + o(1)$.

For the variance, first note that $\mathrm{Var}(c(BG(\pi))) =
H_{n+2}  - \sum_{k=1}^{n+2}\frac{1}{k^2} + o(1).$\\ By further
using the fact that $\log(n+2) = \log(n) + o(1)$ and the
well-known result $\sum_{k=1}^{\infty}\frac{1}{k^2}=
\frac{\pi^2}{6}$, the stated asymptotic formula follows.
\end{proof}

Interestingly, we recover exactly the same asymptotical behaviour
as for the number of cycles in the usual disjoint cycle
decomposition (recall Equations~\eqref{eqn:mean-cycles-asymptotic} and~\eqref{eqn:var-cycles-asymptotic}).

\subsection{The signed case}

We now turn to the problem of computing the expected value and the
variance of the signed Hultman numbers. As in the unsigned case,
we start with the computation of the generating function for the
signed Hultman numbers.
\begin{lemma}\label{gen_fct}
We have
$$G(x)=\sum_{k=1}^{n+1} S_{H}^{\pm}(n,k)x^k=\sum_{\lambda}c_{\lambda}(2)F_{\lambda}(x)F'_{\lambda}(0),
$$
where $\lambda$ is subject to the same restrictions as in
\Cref{formula-for-signed-hultman-numbers} page~\pageref{formula-for-signed-hultman-numbers} and $F_\lambda$ is
defined as in \Cref{F_lambda} page~\pageref{F_lambda}.
\end{lemma}
\begin{proof}
Recall (\Cref{rem} page~\pageref{rem}) that $S_{H}^{\pm}(n,k)$ is
the coefficient of the monomial $h^k\ell$ in the polynomial
$Q_{n+1}^{\mathbb{R}}(h,\ell)$. If we take now $h=x$ and consider
$Q_{n+1}^{\mathbb{R}}(x,\ell)$ as a polynomial only in the
variable $\ell$, we note that the coefficient of the monomial
$\ell$ is obtained by summing up all the terms
$S_{H}^{\pm}(n,k)x^k$, for $k =1,\ldots,n+1$. Therefore, $G(x)$
equals the coefficient of $\ell$ in
$Q_{n+1}^{\mathbb{R}}(x,\ell)$, and hence
\[G(x)=\left.\frac{\partial}{\partial \ell}
Q_{n+1}^{\mathbb{R}}(x,\ell)\right|_{\ell=0}.\]
The formula in the statement easily follows from \Cref{2} page~\pageref{2}.
\end{proof}

In order to compute the expected value and the variance of the
signed Hultman numbers, we will need the following preliminary
lemma.

\begin{lemma}\label{derivatives-F-lambda}
Let $n \geq 1$ and $\lambda$ a partition of $n+1$ of the form
$(a,b,1^{n-a-b+1})$.
\begin{enumerate}
 \item In the case where $a\geq b\geq1$, we have:
\begin{eqnarray*}
F_{\lambda}'(0) &=& \frac{(-1)^{n-a-b}2^{a-b}(a-1)!(2b-2)!(n-a-b+2)!}{(b-1)!},\\
F_{\lambda}'(1) &=& \frac{(-1)^{n-a-b+1}(2a-1)!(b-1)!(n-a-b+1)!}{2^{a-b}(a-1)!},\\
F_{\lambda}''(1) &=&F'_{\lambda}(1)\left\{2H_{2a-1}-2H_{n-a-b+1}-H_{a-1}+H_{b-1}\right\}.
\end{eqnarray*}
\item In the case where $\lambda=(n+1)$, we have:
\begin{eqnarray*}
F_{(n+1)}'(0) &=& 2^nn!,\\
F_{(n+1)}'(1) &=& \frac{(2n+1)!}{2^nn!}(H_{2n+1}-H_{n}/2),\\
F_{(n+1)}''(1) &=&
\frac{(2n+1)!}{2^nn!}\left\{\left(H_{2n+1}-\frac{H_{n}}{2}\right)^2-\sum_{k=0}^n
\frac{1}{(2k+1)^2}\right\}.
\end{eqnarray*}
\end{enumerate}
\end{lemma}

\begin{proof}
We handle both cases separately. 
\begin{enumerate}
 \item Let us first examine the case where $\lambda=(a,b,1^{n+1-a-b})$ and $a \geq b \geq 1$.
In order to simplify the proof, we write $F_{\lambda}(x)=x(x-1)h_{\lambda}(x)$, where
$h_{\lambda}(x)$ is obtained and defined as follows:
\begin{eqnarray*}
 F_{\lambda}(x)&=&2^{a-b}\fallingfactorial{(x/2+a-1)}{a-b}\fallingfactorial{(x+2b-2)}{n+1-a+b}\quad\mbox{(see definition\footnotemark\ page~\pageref{F_lambda})}\\
 &=&2^{a-b}\fallingfactorial{(x/2+a-1)}{a-b}(x+2b-2)(x+2b-1)\cdots (x+1)x (x-1)\\
&&\times (x-2)(x-3)\cdots (x-2+b-n+a)\\
 &=&x(x-1)\underbrace{2^{a-b}\fallingfactorial{(x/2+a-1)}{a-b}\fallingfactorial{(x+2b-2)}{2b-2}\fallingfactorial{(x-2)}{n-a-b+1}}_{=h_{\lambda}(x)}.
\end{eqnarray*}
\footnotetext{Recall, as explained in the statement of \Cref{formula-for-signed-hultman-numbers} page~\pageref{formula-for-signed-hultman-numbers}, that we must replace $n$ with $n+1$.}
\begin{enumerate}
 \item Using the above notation, we have
$$F_{\lambda}'(0) =-h_{\lambda}(0)=
(-1)2^{a-b}(a-1)^{\underline{a-b}}(2b-2)!(-2)^{\underline{n-a-b+1}},$$
from which we easily obtain the wanted expression.

\item We also have
\begin{eqnarray*}
F_{\lambda}'(1) =h_{\lambda}(1)&=& 2^{a-b}\left(a-1/2\right)^{\underline{a-b}}(2b-1)^{\underline{2b-2}}(-1)^{\underline{n-a-b+1}}\\
&=& 2^{a-b}\left(a-1/2\right)^{\underline{a-b}}(2b)!(-1)^{\underline{n-a-b+1}},
\end{eqnarray*}
and obtaining the formula for $F_{\lambda}'(1)$ given in the statement is a simple matter, using the fact that
\begin{eqnarray*}
 \left(a-1/2\right)^{\underline{a-b}}&=&\frac{(2a-1)(2a-3)\cdots (2b+1)}{2^{a-b}}\\
&=&\frac{1}{2^{a-b}}\frac{(2a-1)!}{(a-1)!2^{a-1}}\frac{(b-1)!2^{b-1}}{(2b-1)!}\\
&=&\frac{(2a-1)!b!}{2^{a-b-1}(a-1)!2^{a-b}(2b)!}.
\end{eqnarray*}

\item
In order to simplify the computation of the second derivative, we
will write $F_{\lambda}(x)=(x-1)g_{\lambda}(x)$, where
$$g_{\lambda}(x)=\underbrace{2^{a-b}\left(x/2+a-1
\right)^{\underline{a-b}}}_{=\alpha_{\lambda}(x)}\underbrace{(x+2b-2)^{\underline{2b-1}}}_{=\beta_{\lambda}(x)}\underbrace{(x-2)^{\underline{n-a-b+1}}}_{=\gamma_{\lambda}(x)}.$$
With this notations,
it is easy to see
that
$F_{\lambda}''(1)=2g'_{\lambda}(1)$, with
$$g'_{\lambda}(1)=\alpha'_{\lambda}(1)\beta_{\lambda}(1)\gamma_{\lambda}(1)+\alpha_{\lambda}(1)\beta'_{\lambda}(1)\gamma_{\lambda}(1)+\alpha_{\lambda}(1)\beta_{\lambda}(1)\gamma'_{\lambda}(1).$$
Note that
\begin{eqnarray*}
\alpha'_{\lambda}(1)
&=&\alpha_{\lambda}(1)\left(\frac{1}{2a-1}+\frac{1}{2a-3}+\cdots +
\frac{1}{2b+1}\right)\\
&=&\alpha_{\lambda}(1)\{H_{2a-1}-H_{2b}-(H_{a-1}-H_b)/2\},\\
\beta'_{\lambda}(1) &=&\beta_{\lambda}(1)\sum_{k=1}^{2b-1}
\frac{1}{k}= \beta_{\lambda}(1)H_{2b-1},\\
\gamma'_{\lambda}(1) &=& -\gamma_{\lambda}(1)\sum_{k=1}^{n-a-b+1}
\frac{1}{k} = -\gamma_{\lambda}(1)H_{n-a-b+1},
\end{eqnarray*}
and
\begin{eqnarray*}
\alpha_{\lambda}(1) &=&\frac{(2a-1)!b!}{(2b)!2^{a-b-1}(a-1)!},\\
\beta_{\lambda}(1) &=& (2b-1)!,\\
\gamma_{\lambda}(1) &=& (-1)^{n-a-b+1}(n-a-b+1)!.
\end{eqnarray*}
Combining all of the above, we obtain:
\begin{eqnarray*}
g'_{\lambda}(1)
&=&\alpha_{\lambda}(1)\beta_{\lambda}(1)\gamma_{\lambda}(1)\\
&&\times\{H_{2a-1}-H_{2b}-(H_{a-1}-H_b)/2
+H_{2b-1}-H_{n-a-b+1}\}\\
&=&\frac{(-1)^{n-a-b+1}(2a-1)!(b-1)!(n-a-b+1)!}{2^{a-b}(a-1)!}\\
& &\times \{H_{2a-1}-H_{n-a-b+1}-(H_{a-1}-H_{b-1})/2\}
\end{eqnarray*}
and we finally deduce the formula in the statement.
\end{enumerate}

\item We now turn to the case where $\lambda=(n+1)$, i.e.\ $a=n+1$ and $b=0$.
\begin{enumerate}
\item
Following the definition\footnote{Again, we replace $n$ with $n+1$
in the definition.} of $F_{\lambda}(x)$ given on
page~\pageref{F_lambda}, we have
$$F_{(n+1)}(x)=2^{n+1}\left(x/2+n\right)^{\underline{n+1}}=x \prod_{k=1}^n (x+2k).$$
We thus obtain
$$F_{(n+1)}'(x) = \prod_{k=1}^n (x+2k) + F_{(n+1)}(x) \sum_{k=1}^n \frac{1}{x+2k},$$
which easily gives the wanted expressions when evaluated at $x=0$
and $x=1$.

\item For the second derivative, we obtain
$$F_{(n+1)}''(x) = F_{(n+1)}(x) \sum_{0 \leq i \neq j \leq n}
\frac{1}{(x+2i)(x+2j)},$$ hence
$$F_{(n+1)}''(1) =
\frac{(2n+1)!}{2^nn!}\left\{\left(\sum_{k=0}^n\frac{1}{2k+1}\right)^2-\sum_{k=0}^{n}\frac{1}{(2k+1)^2}\right\},$$
and the formula in the statement follows.
\end{enumerate}
\end{enumerate}
\end{proof}

Knowing
the generating function $G$, we can
easily obtain the expected value of the number of cycles in the
breakpoint graph of a random signed permutation of $n$ elements.

\begin{theorem}\label{mean_signed}
The expected value of the number of cycles in the breakpoint
graph of a uniform random signed permutation $\pi^\pm$ of $n$ elements
is
$$
\mathbb{E}(c(BG(\pi^\pm)))=
H_{2n+1}-\frac{H_{n}}{2} - \sum_{(a,b) \in \mathcal{A}_n}\!\! r_n(a,b),$$
where
$\mathcal{A}_n
=
\{(a,b) \in \mathbb{N}^2: \ a
\geq b \geq 1, \ a+b \leq n+1\}$ and
$$r_n(a,b)
= \frac{(-1)^{n+a-b}(n+1)(2a-2b+1)(a-1)!(2b-2)!(n-a-b+2)!}{2^{n-a+b-1}n!(b-1)!(n+a-b+2)^{\underline{2}}
(n-a+b+1)^{\underline{2}} }.
$$
\end{theorem}

\begin{proof}
As recalled in the proof of
\Cref{mean-of-unsigned-hultman-numbers}, we have $\mathbb{E}(c(BG(\pi^\pm))) =
G'(1)/G(1)$. Note that, by definition,
$G(1)=\sum_{k=1}^{n+1} S_{H}^{\pm}(n,k)$,
which equals the
number of signed permutations of $n$ elements, i.e.\ $2^nn!$. By
\Cref{gen_fct}, the expected number of cycles in the breakpoint
graph of a random signed permutation is
\begin{equation*}
\mathbb{E}(c(BG(\pi^\pm)))=\frac{1}{2^nn!}\sum_{\lambda}
c_{\lambda}(2)F'_{\lambda}(1)F'_{\lambda}(0).
\end{equation*}

Using the formulas for $F'_{\lambda}(1)$ and $F'_{\lambda}(0)$
derived in \Cref{derivatives-F-lambda} and the expression for the
coefficients\footnote{Again, we replace $n$ with $n+1$
in the definitions.} $c_{\lambda}(2)$ given in \Cref{c_lambda1,c_lambda2} page~\pageref{c_lambda1},
the formula in the statement follows.
\end{proof}

The generating function $G$
allows us also to compute the
variance of the signed Hultman numbers.
\begin{theorem}
The variance of the number of cycles in the breakpoint graph of a
uniform random signed permutation $\pi^\pm$ of $n$ elements is
\begin{align*}
& \mathrm{Var} (c(BG(\pi^\pm)))
 = \ H_{2n+1}-\frac{H_{n}}{2}
-\sum_{k=0}^n\frac{1}{(2k+1)^2} - \left(\sum_{(a,b) \in \mathcal{A}_n}\!\! r_n(a,b) \right)^2\\
& + \sum_{(a,b) \in \mathcal{A}_n}\!\! r_n(a,b)\{2H_{2n+1}-H_n -
2H_{2a-1} +2H_{n-a-b+1}+H_{a-1}-H_{b-1} - 1\},
\end{align*}
where $\mathcal{A}_n$ and the coefficients $r_n(a,b)$ are as defined
in \Cref{{mean_signed}}.
\end{theorem}

\begin{proof}
As recalled in the proof of
\Cref{variance-of-unsigned-hultman-numbers}, the variance
can be
obtained from the generating function $G$ by evaluating the
function $(\log G)'(x)+(\log G)''(x)$ at $x=1$. Therefore, the
variance of the number of cycles in the breakpoint graph of a
random signed permutation equals
\begin{eqnarray*}
&&\frac{G'(1)}{G(1)}+\frac{G''(1)}{G(1)}-\left(
\frac{G'(1)}{G(1)}\right)^2\\
&=&\frac{G'(1)+G''(1)}{G(1)}-(\mathbb{E}(c(BG(\pi^\pm))))^2\\
&=&\frac{1}{2^nn!}\sum_{\lambda}
c_{\lambda}(2)(F'_{\lambda}(1)+F''_{\lambda}(1))F'_{\lambda}(0)-(\mathbb{E}(c(BG(\pi^\pm))))^2.\quad\mbox{(using \Cref{gen_fct})}
\end{eqnarray*}

Using
the formulas for $F'_{\lambda}(1),
F''_{\lambda}(1)$ and $F'_{\lambda}(0)$ given in
\Cref{derivatives-F-lambda}, we obtain that the variance equals
\begin{align*}
& H_{2n+1}-\frac{H_{n}}{2} -\sum_{k=0}^n\frac{1}{(2k+1)^2}+
\left(H_{2n+1}-\frac{H_{n}}{2}\right)^2 - (\mathbb{E}(c(BG(\pi^\pm))))^2\\
& - \sum_{(a,b) \in \mathcal{A}_n}\!\!
r_n(a,b)\left\{2H_{2a-1}-2H_{n-a-b+1}-H_{a-1}+H_{b-1} + 1\right\},
\end{align*}
which equals the wanted expression once
$\mathbb{E}(c(BG(\pi^\pm)))$ is replaced with the value derived in \Cref{mean_signed}.
\end{proof}

As in the unsigned case, we will study the behaviour of the mean
and variance for large
values of
 $n$. To that end, we will first prove
the following lemma.

\begin{lemma} As $n \to \infty$, we have
$$\sum_{(a,b) \in \mathcal{A}_n}\!\! |r_n(a,b)| =
\frac{1}{\log(n)}\times o(1).$$
\end{lemma}
\begin{proof}
If we denote $k=a-b$, the above sum becomes
\begin{eqnarray*}
&& \sum_{k=0}^{n-1}
\frac{2^{k-n+1}(n+1)(2k+1)}{n!(n+k+2)^{\underline{2}}
(n-k+1)^{\underline{2}}}\!\! \sum_{b=1}^{\lfloor(n-k+1)/2\rfloor}
\frac{(k+b-1)!(2b-2)!(n-k-2b+2)!}{(b-1)!}\\
&=& \sum_{k=0}^{n-1}
\frac{2^{k-n+1}(n+1)(2k+1)}{(n+k+2)^{\underline{2}}
(n-k+1)(k+1)\binom{n}{k+1}}\!\!
\sum_{b=1}^{\lfloor(n-k+1)/2\rfloor}
\frac{\binom{k+b-1}{k}}{\binom{n-k}{2b-2}}\\
&\leq& \sum_{k=0}^{n-1} \frac{2^{k-n+1}}{(n+k+2)\binom{n}{k+1}}\!\!
\sum_{b=1}^{\lfloor(n-k+1)/2\rfloor}\binom{k+b-1}{k}\\
&=& \sum_{k=0}^{n-1}
\frac{2^{k-n+1}\binom{k+\lfloor(n-k+1)/2\rfloor}{k+1}}{(n+k+2)\binom{n}{k+1}}.\quad\quad\quad\quad\quad\mbox{(using $\sum_{j=k}^{n} \binom{j}{k}= \binom{n+1}{k+1}$)}
\end{eqnarray*}
We further
observe
that
$$\sum_{(a,b) \in \mathcal{A}_n}\!\! |r_n(a,b)| \leq \sum_{k=0}^{n-1}
\frac{2^{k-n+1}}{n+k+2} \leq
2\left(1-\frac{1}{2^n}\right)\frac{1}{n+2},$$ and the result in
the statement easily follows.
\end{proof}

Based on this lemma, we can now obtain the following.
\begin{remark} When $n \to \infty$, the expected value and variance of
the number of cycles in the breakpoint graph of a uniform random
signed permutation $\pi^\pm$ of $n$ elements have the following
asymptotical behaviour:
\begin{align*}
\mathbb{E}(c(BG(\pi^\pm))) &= \frac{\log(n)}{2} + \frac{\gamma}{2} + \log(2) + o(1) ,\\
\mathrm{Var}(c(BG(\pi^\pm))) & = \frac{\log(n)}{2} + \frac{\gamma}{2}
+ \log(2) - \frac{\pi^2}{8}+ o(1).
\end{align*}
\end{remark}

Note that, in the limit when $n \to \infty$, the mean and variance
in the signed case are of the same order $(\log(n))$ as in the
unsigned case, but they differ by a factor of $1/2$.

\section{Applications: Distributions of rearrangement distances}\label{sec:distributions}

As stated in the introduction of this paper, the breakpoint graph and its cycles are used in a lot of variants of genome rearrangement problems to compute evolutionary distances -- either exactly or approximately. In this section, we are interested in exploring to what extent we can rely on those cycles in order to approximate the distribution of several distances that have been studied in the field of genome rearrangements, so as to obtain a better idea of how tight a particular bound on a distance is, or whether it is worth computing a distance exactly in cases where this requires solving an \NP-hard problem. By ``distribution of a distance'', we mean the number of (possibly signed) permutations of $n$ elements whose distance equals $k$, for all possible values of $k$.

We will not say much about rearrangement distances or how to compute them, except for the fact that, as already stated earlier in this paper, they are based on a set $S$ of operations that generate $S_n$ (resp.\ $S^\pm_n$). In the following, what we mean by expressions like ``the $S$ distance of $\pi$'' is the minimum number of operations from $S$ needed to transform a given permutation $\pi$ into the identity permutation $\iota$; a few examples of such operations that we will consider here are summarised informally in \Cref{tab:distances}. The reader should bear in mind that the discussion presented in this section focuses on experiments with relatively small amounts of data (mainly because many interesting distances are hard to compute, and because the number of (signed) permutations grows much too fast to generate the full distributions for large values of $n$), which is why we refrain from making any bold conjecture or actually proving any result.
We will also restrict ourselves to comparing distributions for one fixed value of $n$, namely, the largest value for which we could obtain the distribution of the particular distance we are interested in; similar-looking plots can however be obtained for any value. We generated the distributions based on cycles of the breakpoint graph ourselves, but the distributions of the distances we consider here were computed by \citet{galvao-db}.

\subsection{Unsigned distances}

\begin{table}
\centering
\scalebox{0.79}{
  \begin{tabular}{|l|l|l|l|}
  \hline
  & \textbf{Distance} & \textbf{Operation} & \textbf{Description of the operation} \\
  \hline
\multirow{5}{*}{\rotatebox{90}{\scalebox{0.7}{unsigned}}}  & $bid$ & block-interchange & exchanges two non-necessarily adjacent segments \\
 & $td$ & transposition & exchanges two adjacent segments \\
 & $ptd$ & prefix transposition & transposition involving $\pi_1,\pi_2,\ldots, \pi_k$ for some $k$ \\
 & $rd$ & reversal & reverses a segment \\
 & $prd$ & prefix reversal & reversal involving $\pi_1,\pi_2,\ldots, \pi_k$ for some $k$ \\
  \hline
\multirow{2}{*}{\rotatebox{90}{\scalebox{0.7}{signed}}}  & $srd$ & signed reversal & reverses a segment and flips the signs in that segment\\
 & $psrd$ & prefix signed reversal & signed reversal involving $\pi_1,\pi_2,\ldots, \pi_k$ for some $k$ \\
  \hline
  \end{tabular}
}
\caption{Some abbreviations and informal definitions used throughout this section.}
\label{tab:distances}
\end{table}

A few distances between unsigned permutations have been considered in the field of genome rearrangements~\cite{fertin-combinatorics}. \citet{doignon-hultman} already observed that $\hultmannumber{n}{n+1-2k}$ is exactly the number of permutations $\pi$ in $S_n$ whose \emph{block-interchange distance} $bid(\pi)$ equals $k$, an immediate consequence of the following result.
\begin{theorem}\label{thm:bid}
\cite{christie-block} For all $\pi$ in $S_n$, we have $bid(\pi)=(n+1-c(BG(\pi)))/2$.
\end{theorem}
Whereas sorting by block-interchanges and computing $bid(\pi)$ can be achieved in polynomial time~\cite{christie-block}, this is not the case for any of the other unsigned operations listed in \Cref{tab:distances}: sorting by transpositions and sorting by reversals, as well as computing the related distances, are \NP-hard problems (see \citet{bulteau-sbt-is-hard} and \citet{caprara-sorting}, respectively); the same problems in the context of prefix reversals are also \NP-hard~\cite{bulteau-pf-is-hard}, while their complexity in the case of prefix transpositions is open.

However, since transpositions are but a particular case of block-interchanges, the expression given in \Cref{thm:bid} for computing $bid(\pi)$ is also a lower bound on the \emph{transposition distance} $td(\pi)$. Additionnally, a tighter lower bound on the transposition distance was proved by \citet{bafna-transpositions}.

\begin{theorem}\label{thm:td-lb}
\cite{bafna-transpositions} For all $\pi$ in $S_n$, we have $td(\pi)\geq(n+1-c_{odd}(BG(\pi)))/2$, where $c_{odd}(BG(\pi))$ is the number of cycles of odd length in $BG(\pi)$.
\end{theorem}

Consequently, it makes sense to try to approximate the distribution of the transposition distance using $\hultmannumber{n}{n+1-2k}$ (because of \Cref{thm:bid}) and what could be called the \emph{odd} Hultman numbers $S_H^{odd}(n,n+1-2k)$, i.e.\ the number of permutations of $n$ elements whose breakpoint graph contains $n+1-2k$ cycles of odd length (because of \Cref{thm:td-lb}). \Cref{plot:td-versus-shn}$(a)$ compares all three distributions for $n=13$. To the best of our knowledge, there is no known formula for computing odd Hultman numbers.

\citet{dias-prefix} initiated the study of \emph{prefix transpositions}, which are transpositions that can only be applied to an initial segment of the permutation to sort. To the best of our knowledge, the complexity of sorting by prefix transpositions or computing the corresponding distance is still open.
However, a lower bound on the prefix transposition distance based on the breakpoint graph is known.

\begin{theorem}\label{thm:my-ptd-lower-bound-i}
\cite{labarre-edit} For any $\pi$ in $S_n$, we have
\begin{eqnarray}\label{eqn:my-ptd-lower-bound-i}
ptd(\pi)\geq\frac{n+1+c(BG(\pi))}{2}-c_1(BG(\pi))-\left\{
\begin{array}{ll}
0 & \mbox{if } \pi_1= 1, \\
1 & \mbox{otherwise},
\end{array}
\right.
\end{eqnarray}
where $c_1(BG(\pi))$ is the number of cycles of length $1$ in $BG(\pi)$.
\end{theorem}

\Cref{plot:td-versus-shn}$(b)$ shows the distribution of the prefix transposition distance, together with some function of the Hultman numbers and the distribution of the number of permutations in $S_n$ for which lower bound~\eqref{eqn:my-ptd-lower-bound-i} equals $k$ for $n=13$. On this particular plot and the forthcoming ones, we find the offset $m$ in $\hultmannumber{n}{n+1-k+m}$ experimentally by shifting the distribution of $\hultmannumber{n}{n+1-k}$ so that it best fits the distribution of the distance we are interested in.

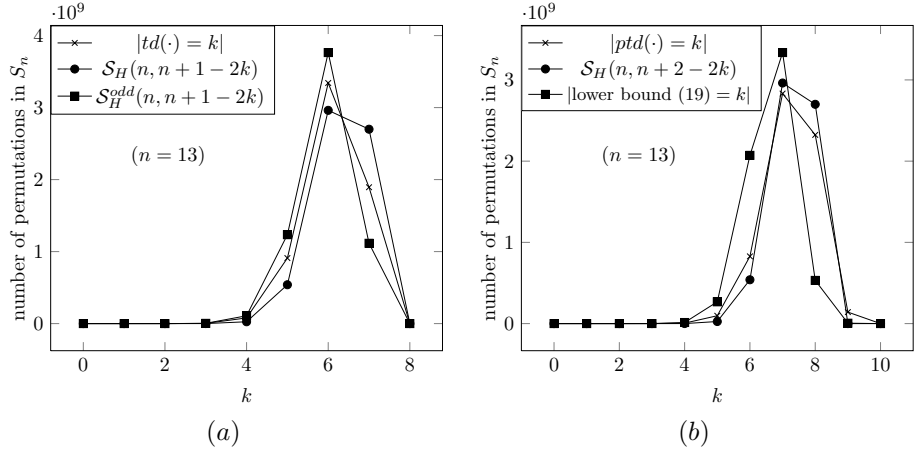
\begin{figure}[htbp]
\pgfplotsset{every axis legend/.append style={
at={(0,1)},
anchor=north west}}
\centering
\begin{tabular}{cc}
\begin{tikzpicture}[scale=0.76]
\pgfplotsset{every axis/.append style={
extra description/.code={
\node at (0.3,0.6) {$(n=13)$};
}}}
\begin{axis}[
  xlabel=$k$,
  ylabel=number of permutations in $S_n$,
  y label style={at={(0.1,0.5)}}]
\pgfplotstableread{td-versus-cycles-13.dat.tex}\table
\addplot[mark=x,color = black]       table[x=n,y=td] from \table;
\addplot[mark=*,color = black]       table[x=n,y=inv] from \table;
\addplot[mark=square*,color = black] table[x=n,y=odd_inv] from \table;
\legend{$|td(\cdot)=k|$,$\hultmannumber{n}{n+1-2k}$,$\oddhultmannumber{n}{ n+1-2k }$}
\end{axis}
\end{tikzpicture}
&
\begin{tikzpicture}[scale=0.76]
\pgfplotsset{every axis/.append style={
extra description/.code={
\node at (0.3,0.6) {$(n=13)$};
}}}
\begin{axis}[
  xlabel=$k$,
  ylabel=number of permutations in $S_n$,
  y label style={at={(0.1,0.5)}}  ]
\pgfplotstableread{ptd-versus-cycles-13.dat.tex}\table
\addplot[mark=x,color = black]        table[x=n,y=ptd] from \table;
\addplot[mark=*,color = black]        table[x=n,y=inv] from \table;
\addplot[mark=square*,color = black]  table[x=n,y=labarre_ptd_lb] from \table;
\legend{$|ptd(\cdot)=k|$,$\hultmannumber{n}{n+2-2k}$,\small $|\mbox{lower
bound~\eqref{eqn:my-ptd-lower-bound-i}}=k|$}
\end{axis}
\end{tikzpicture}
\\
$(a)$ & $(b)$
\end{tabular}
\caption{$(a)$ How the distributions of the unsigned and odd Hultman numbers relate to the distribution of the transposition distance, for $n=13$; $(b)$ how the distributions of the unsigned Hultman numbers and the number of permutations for which lower bound~\eqref{eqn:my-ptd-lower-bound-i} equals $k$ relate to the distribution of the prefix transposition distance, for $n=13$.}
\label{plot:td-versus-shn}
\end{figure}

Two other distances that have received a considerable amount of attention are the \emph{reversal distance}, where a reversal reverses the order of the elements contained in the segment of the permutation on which it acts, and the \emph{prefix reversal distance}, where prefix reversals have the same effect as reversals but may only be applied to an initial segment of the permutation. \citet{caprara-sorting} showed that computing the former is \NP-hard, while \citet{bulteau-pf-is-hard} proved that computing the latter is \NP-hard. Again, we find it interesting to examine how the distribution of the number of cycles in the breakpoint graph relates to those distances, which we do in \Cref{fig:reversals-and-hultman-numbers}. We warn the reader familiar with breakpoint graphs, however, that the breakpoint graph used in our paper differs from the structure traditionally used for the study of these two distances, which admits more than one
cycle decomposition; the graph we use can be
seen as the result
of selecting one particular decomposition among all possible decompositions. In this setting, there is a much larger difference between the distributions of both distances and of the unsigned Hultman numbers than what we have observed for transpositions in \Cref{plot:td-versus-shn}, which confirms that using only (our version of) the breakpoint graph in this case is not enough.

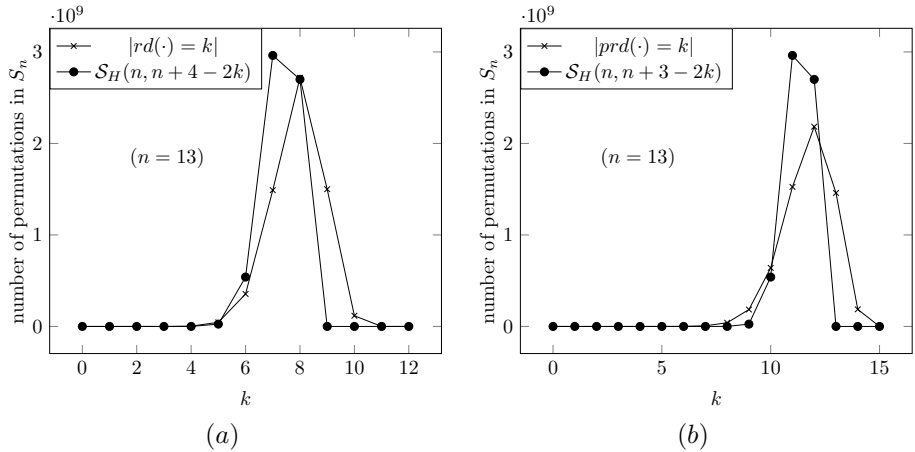
\begin{figure}[htbp]
\pgfplotsset{every axis legend/.append style={
at={(0,1)},
anchor=north west}}
\pgfplotsset{every axis/.append style={
extra description/.code={
\node at (0.3,0.6) {$(n=13)$};
}}}
\centering
\begin{tabular}{cc}
\begin{tikzpicture}[scale=0.76]
\begin{axis}[
  xlabel=$k$,
  ylabel=number of permutations in $S_n$,
  y label style={at={(0.1,0.5)}}  ]
\pgfplotstableread{rd-versus-cycles-13.dat.tex}\table
\addplot[mark=x,color = black]  table[x=n,y=rd] from \table;
\addplot[mark=*,color = black]  table[x=n,y=inv] from \table;
\legend{$|rd(\cdot)=k|$,$\hultmannumber{n}{n+4-2k}$}
\end{axis}
\end{tikzpicture}
&
\begin{tikzpicture}[scale=0.76]
\begin{axis}[
  xlabel=$k$,
  ylabel=number of permutations in $S_n$,
  y label style={at={(0.1,0.5)}}  ]
\pgfplotstableread{prd-versus-cycles-13.dat.tex}\table
\addplot[mark=x,color = black]  table[x=n,y=prd] from \table;
\addplot[mark=*,color = black]  table[x=n,y=inv] from \table;
\legend{$|prd(\cdot)=k|$,$\hultmannumber{n}{n+3-2k}$}
\end{axis}
\end{tikzpicture}
\\
$(a)$ & $(b)$
\end{tabular}
\caption{How the distribution of the unsigned Hultman numbers relates to the distribution of $(a)$ the reversal distance and $(b)$ the prefix reversal distance, for $n=13$.}
\label{fig:reversals-and-hultman-numbers}
\end{figure}

\subsection{Signed distances}

A number of well-studied and biologically relevant distances between signed permutations are also based on the breakpoint graph. These include the \emph{double cut-and-join (DCJ) distance}, introduced by \citet{yancopoulos-efficient}, who showed that its value could be computed using the formula $dcj(\pi)=n+1-c(BG(\pi))$. As a consequence, the number of signed permutations of $n$ elements with DCJ distance $k$ is exactly $\signedhultmannumber{n}{n+1-k}$.

Another distance whose distribution can be well approximated using the signed Hultman numbers is the \emph{signed reversal distance} (see \Cref{tab:distances} for an informal definition of signed reversals).
\citet{hannenhalli-transforming} proved the following formula for
computing the signed reversal distance of any permutation $\pi$, denoted by $srd(\pi)$.

\begin{theorem}
\cite{hannenhalli-transforming} For any $\pi$ in $S_n^\pm$, the signed reversal distance of $\pi$ is
$$srd(\pi)=n+1-c(BG(\pi))+h(\pi)+f(\pi),$$
where $h(\pi)$ is the number of ``hurdles'' of $\pi$ and $f(\pi)=1$ if $\pi$ is a ``fortress'', and $0$ otherwise.
\end{theorem}

We will not give more details on the terms ``hurdles'' and ``fortress'' (see \citet{hannenhalli-transforming} for definitions), except for the fact that hurdles are particular collections of cycles in $BG(\pi)$, and that a permutation cannot be a fortress unless $h(\pi)>0$. Our point here is that the following lower bound, first proved by
\citet{bafna-genome-journal}, is extremely tight:
\begin{equation}\label{lower_bound}
\forall\ \pi\in S^\pm_n: srd(\pi) \geq n+1-c(BG(\pi)).
\end{equation}

This claim is supported by \citeauthor{caprara}'s proof~\cite{caprara} of the fact that the probability that a permutation $\pi\in S^\pm_n$ is \emph{not} tight with respect to \Cref{lower_bound} is $\Theta(n^{-2})$, and by \citeauthor{swenson}'s proof~\cite{swenson} that the probability that $\pi$ is a fortress is $\Theta(n^{-15})$.
Therefore, \Cref{lower_bound} provides a very good approximation
of the signed reversal distance, and the distribution of $\signedhultmannumber{n}{n+1-k}$ closely matches that of the signed reversal distance. \Cref{plot:srd-versus-shn} illustrates the situation for the case $n=10$.

\begin{figure}[htbp]
\centering
\begin{tikzpicture}
\pgfplotsset{every axis/.append style={
extra description/.code={
\node at (0.5,0.7) {$(n=10)$};
}}}
\begin{axis}[
  xlabel=$k$,
  ylabel=number of permutations in $S^\pm_n$,
  y label style={at={(0.05,0.5)}}  ]
\pgfplotstableread{srd-versus-cycles-10.dat.tex}\table
\addplot[ycomb,mark=x]  table[x=n,y=srd] from \table;
\addplot[ycomb,mark=triangle]  table[x=n,y=inv] from \table;
\pgfplotsset{every axis legend/.append style={
at={(0,1)},
anchor=north west}}
\legend{$|srd(\cdot)=k|$,$\signedhultmannumber{n}{n+1-k}$}
\end{axis}
\end{tikzpicture}
\caption{The distributions of the signed reversal distance and of the signed Hultman numbers, for $n=10$.}
\label{plot:srd-versus-shn}
\end{figure}
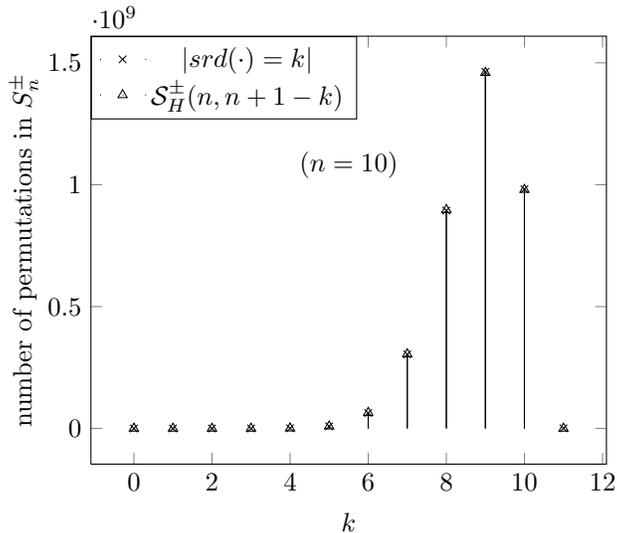

Other distances have not been studied with that level of detail, which is why we find it interesting to try to relate their distribution to that of the Hultman numbers. A particular restriction of the signed reversal distance is the \emph{prefix} signed reversal distance, denoted by $psrd(\cdot)$, whose definition follows that of the signed reversal distance except that reversals can only act on an initial segment of the permutation. No formula is known for computing that distance, and the computational complexity of the problem has remained open since the first works on the subject~\cite{cohen-burnt}. However, a lower bound based on the breakpoint graph was recently obtained by \citet{labarre-burnt-pancakes}, which naturally prompts us to wonder how exactly we can rely on the breakpoint graph to approximate that distance.

\begin{theorem}\label{thm:lower-bound-on-psrd}
\cite{labarre-burnt-pancakes} For any $\pi$ in $S^\pm_n$, we have
\begin{eqnarray}\label{eqn:lower-bound-on-psrd}
 psrd(\pi)\geq
n+1+c(BG(\pi))-2c_1(BG(\pi))-\left\{
\begin{array}{ll}
0 & \mbox{if } \pi_1= 1, \\
2 & \mbox{otherwise}.
\end{array}
\right.
\end{eqnarray}
\end{theorem}

\Cref{plot:psrd-versus-shn} shows a plot with the distribution of the prefix signed reversal distance and that of the signed Hultman numbers, as well as of the distribution of lower bound~\eqref{eqn:lower-bound-on-psrd} for $n=10$. It can be seen on that graph that the
latter is quite far off from the distribution of the prefix signed reversal distance, hinting that additional work seems needed to reduce the gap between the lower bound and the actual distance.

\begin{figure}[htbp]
\pgfplotsset{every axis legend/.append style={
at={(1.05,1)},
anchor=north west}}
\pgfplotsset{every axis/.append style={
extra description/.code={
\node at (0.3,0.6) {$(n=10)$};
}}}
\centering
\begin{tikzpicture}
\begin{axis}[
  xlabel=$k$,
  ylabel=number of permutations in $S^\pm_n$,
  y label style={at={(0.05,0.5)}}  ]
\pgfplotstableread{psrd-versus-cycles-10.dat.tex}\table
\addplot[mark=x,color = black]        table[x=n,y=psrd] from \table;
\addplot[mark=*,color = black]        table[x=n,y=inv] from \table;
\addplot[mark=square*,color = black]  table[x=n,y=labarre_cibulka_lb] from
\table;
\legend{$|psrd(\cdot)=k|$,$\signedhultmannumber{n}{n+4-k}$,$|$lower bound~\eqref{eqn:lower-bound-on-psrd} $=k|$}
\end{axis}
\end{tikzpicture}
\caption{The distributions of the prefix signed reversal distance, of the signed Hultman numbers, and of the number of permutations for which lower bound~\eqref{eqn:lower-bound-on-psrd} equals $k$, for $n=10$.}
\label{plot:psrd-versus-shn}
\end{figure}
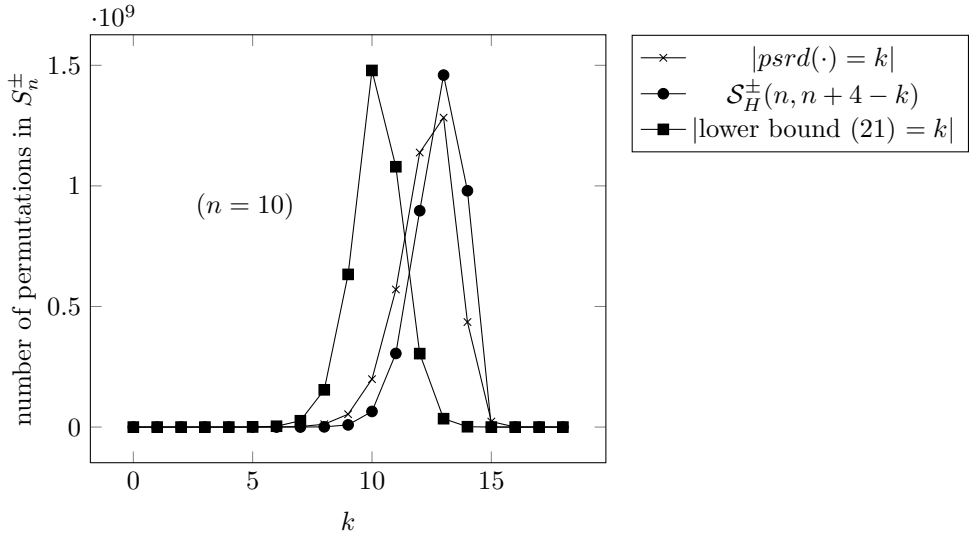

\section{Conclusions}

In this paper, we proved the first explicit formula for
enumerating signed permutations whose breakpoint graph contains a
given number of cycles, and proved simpler expressions for
particular cases. We also obtained a new expression for
enumerating unsigned permutations whose breakpoint graph contains
a given number of cycles, and used both formulas to derive simpler
proofs of some other previously known results. Getting more
insight into breakpoint graphs and their cycle decomposition is
particularly relevant to edit distances used in the field of
genome rearrangements, and we hope that our results can help shed
light on their distributions, expected values and variances. There
are several interesting directions in which our work could be
extended, which we outline and motivate below.

Just like one can define conjugacy classes in the symmetric and hyperoctahedral groups, we could investigate conjugacy classes with respect to the breakpoint graph. This was already initiated by \citet{doignon-hultman}, who referred to them as ``Hultman classes'' and provided explicit formulas for enumerating those classes in the case of unsigned permutations. More work remains to be done in the unsigned case: indeed, the work done by \citet{bona-average} provides us with a very nice formula for computing the distribution of cycles, but no simpler expression than the complicated ones obtained by \citet{doignon-hultman} is yet known for enumerating Hultman classes or their cardinalities. Moreover, no work so far has been done in order to enumerate Hultman classes in the signed setting, and obtaining an expression for enumerating the so-called ``simple permutations'',
which are defined in this context as permutations
whose breakpoint graph contains no cycle of length greater than $2$, seems especially interesting (for more information about the importance
of those permutations in genome rearrangements, see \citet{hannenhalli-transforming} and \citet{labarre-burnt-pancakes}).

The expression we obtained for the signed Hultman numbers is quite useful in practice, since it allows us to obtain the distribution of those numbers for large values of $n$.
Unfortunately, it does not seem easy to use in order to gain insights and have an intuitive interpretation of the shape of the distribution,
which would be useful in order to know how this distribution can be approximated or how it grows as $n$ increases.
Finding simpler generating functions, recurrence relations
or nicer
formulas would be useful in that regard and in order to obtain
more information on the properties of this distribution.

The connection between the cycle structure of breakpoint graphs
and factorisations of even permutations
(\Cref{unsigned-hultman-counts-factorisations},
page~\pageref{unsigned-hultman-counts-factorisations}) proved
useful not only in characterising the distribution of those cycles
and of the related cycle types, but also provided the foundations
of a simple and generic method for obtaining lower bounds on
\emph{any} ``revertible'' edit distance between unsigned
permutations (see \citet{labarre-edit} for more details). Is there
any way to use the results and connections obtained in
\Cref{sec:formula-signed-hultman} in order to obtain similar
results for signed permutations?

Finally, recall that permutations are just one way of modelling genomes. One natural direction would be to investigate the distribution of cycles in the breakpoint graph of other structures, like set systems or ``fragmented'' permutations (see again \citet{fertin-combinatorics} for an overview of existing models).

\section{Acknowledgements}

The first author was partially supported by the ANR MAEV under
contract ANR-06-BLAN-0113. Both authors also wish to thank the
group ``Evolution Biologique et Mod\'elisation'', LATP, Universit\'e
de Provence, where part of this research was performed, as well as Mathilde
Bouvel for bringing reference~\cite{kwak-genus} to their attention.

\bibliographystyle{elsarticle-num-names}
\bibliography{signed-hultman-numbers}

\end{document}